\documentclass[11pt]{article}%
\usepackage{amsfonts}
\usepackage{amsmath}
\usepackage{amssymb}
\usepackage{amsxtra}
\usepackage{graphicx}
\usepackage{geometry}
\usepackage{color}
\usepackage{colortbl}%
\providecommand{\U}[1]{\protect\rule{.1in}{.1in}}
\newtheorem{theorem}{Theorem}

\newtheorem{lemma}[theorem]{Lemma}

\newtheorem{proposition}[theorem]{Proposition}
\newtheorem{remark}[theorem]{Remark}

\newenvironment{proof}[1][Proof]{\noindent\textbf{#1.} }{\ \rule{0.5em}{0.5em}}
\numberwithin{equation}{section}
\ifx\pdfoutput\relax\let\pdfoutput=\undefined\fi
\newcount\msipdfoutput
\ifx\pdfoutput\undefined\else
\ifcase\pdfoutput\else
\ifx\paperwidth\undefined\else
\ifdim\paperheight=0pt\relax\else\pdfpageheight\paperheight\fi
\ifdim\paperwidth=0pt\relax\else\pdfpagewidth\paperwidth\fi
\fi\fi\fi
\begin{document}

\title{Thin elastic plates supported over small areas. \\I. Korn's inequalities and boundary layers}
\author{G. Buttazzo\thanks{Universit\`{a} di Pisa, Department of Mathematics, Largo B.
Pontecorvo, 5, 56127 Pisa, Italy; email: buttazzo@dm.unipi.it.},
G.Cardone\thanks{Universit\`{a} del Sannio, Department of Engineering, Corso
Garibaldi, 107, 82100 Benevento, Italy; email: giuseppe.cardone@unisannio.it.}%
, S.A.Nazarov\thanks{Mathematics and Mechanics Faculty, St. Petersburg State
University, 198504, Universitetsky pr., 28, Stary Peterhof, Russia;
Saint-Petersburg State Polytechnical University, Polytechnicheskaya ul., 29,
St. Petersburg, 195251, Russia; Institute of Problems of Mechanical
Engineering RAS, V.O., Bolshoj pr., 61, St. Petersburg, 199178, Russia; email:
srgnazarov@yahoo.co.uk.}}
\maketitle

\begin{abstract}
\medskip A thin anisotropic elastic plate clamped along its lateral side and
also supported at a small area $\theta_{h}$ of one base is considered; the
diameter of $\theta_{h}$ is of the same order as the plate relative thickness
$h\ll1$. In addition to the standard Kirchhoff model with the Sobolev point
condition, a three-dimensional boundary layer is investigated in the vicinity
of the support $\theta_{h}$, which with the help of the derived weighted
inequality of Korn's type, will provide an error estimate with the bound
$ch^{1/2}|\ln h|$. Ignoring this boundary layer effect reduces the precision
order down to $|\ln h|^{-1/2}$.



\end{abstract}

\section{Introduction\label{sect1}}

\subsection{A plate supported over a small area\label{sect1.1}}

Let $\omega$ be a domain in the plane $\mathbb{R}^{2}$ with a smooth boundary
$\partial\omega$ and a compact closure $\overline{\omega}=\omega\cup
\partial\omega$. We introduce the cylindrical plate
\begin{equation}
\Omega_{h}=\omega\times\left(  -h/2,h/2\right)  \label{1.1}%
\end{equation}
of a small thickness $h>0$. By rescaling, the characteristic size of $\omega$
is reduced to one so that the Cartesian coordinates $x=(y,z) \in\mathbb{R}%
^{2}\times\mathbb{R}$ and all geometric parameters become dimensionless. The
bases of the plate and its lateral side are given by%
\begin{equation}
\Sigma_{h}^{\pm}=\left\{  x:y=\left(  y_{1},y_{2}\right)  \in\omega,\ z=\pm
h/2\right\}  ,\ \ \ \upsilon_{h}=\left\{  x:y\in\partial\omega,\ \left\vert
z\right\vert <h/2\right\}  \label{1.2}%
\end{equation}
respectively. We fix a point $\mathcal{O}$ inside $\omega$ and place the
$y$-coordinate origin at $\mathcal{O}$. Denoting by $\theta\subset
\mathbb{R}^{2}$ an open, not necessarily connected, set with a compact closure
$\overline{\theta}$, we assume that the plate (\ref{1.1}) is clamped over the
lateral side $\upsilon_{h}$ as well as at the small supporting area%
\begin{equation}
\theta_{h}=\left\{  x:\eta:=h^{-1}y\in\theta,\ z=-h/2\right\}  \label{1.3}%
\end{equation}
on the base $\Sigma_{h}^{-}$. Characteristic sizes of $\theta$ are supposed to
be of order one, too. In other words, the diameter of the supporting zone
(\ref{1.3}) is comparable with the plate thickness $h$. The bound $h_{0}>0$
for the small parameter $h$ is chosen such that $\overline{\theta_{h}}
\subset\Sigma_{h}^{-}$ for all $h\in(0,h_{0}]$, however if necessary, we will
reduce $h_{0}$ but always keep the notation. Throughout this paper we do not
distinguish in the notation for $\theta$ and $\theta_{h}$ between
two-dimensional sets and their immersions in $\mathbb{R}^{3}$ along the planes
$\left\{  x:z=-1/2\right\}  $ and $\left\{  x:z=-h/2\right\}  $, respectively.

The plate $\Omega_{h}$ is made out of a homogeneous anisotropic elastic
material and its deformation is caused by volume forces. It is ideally fixed
over the set
\begin{equation}
\Gamma_{h}=\upsilon_{h}\cup\theta_{h} \label{1.16}%
\end{equation}
while the rest of the plate surface, in particular $\Sigma_{h}^{\bullet
}=\Sigma_{h}^{-}\setminus\overline{\theta_{h}}$, is traction-free. The clamped
area is shaded in Figure \ref{f1}.

\begin{figure}[ptb]
\begin{center}
\includegraphics[scale=0.85]{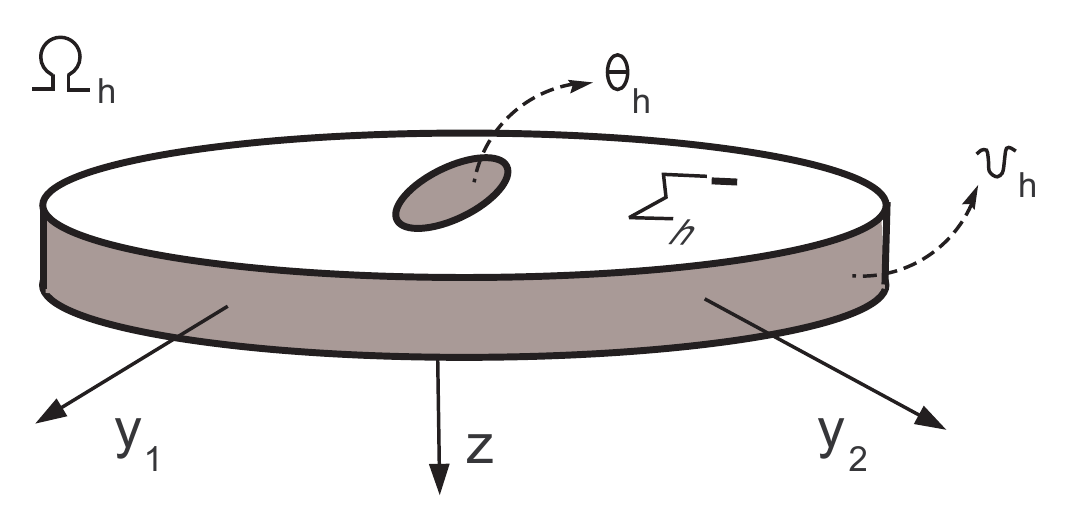}
\end{center}
\caption{A plate clamped over the lateral side $\upsilon_{h}$ and a small
support area $\theta_{h}$.}%
\label{f1}%
\end{figure}

The main goal of this and consequent paper \cite{ButCaNa2} is to examine the
influence of the supporting area (\ref{1.3}) on the stress-strain state of the
whole plate. To this end, we construct asymptotics of elastic fields as
$h\to+0$, prove error estimates, and create a two-dimensional model which
reflects adequately all principal effects of the small support. We emphasize
that clamping along the lengthy set $\upsilon_{h}\subset\partial\Omega_{h}$
plays a considerable role in technicalities of our study and the case of a
traction-free lateral side will be treated in a subsequent paper where, in
contrast to the present case, the number of small supporting areas and their
location become of a major importance.

\subsection{Formulation of the elasticity problem; the Mandel-Voigt
notation\label{sect1.2}}

Since the Cartesian coordinate system $x=\left(  x_{1},x_{2},x_{3}\right)  $
attached to the plate $\Omega_{h}$ is fixed, we can regard the displacement
field $u$ as the column $\left(  u_{1},u_{2},u_{3}\right)  ^{\top}$ in
$\mathbb{R}^{3}$ where $u_{j}$ is the projection of $u$ onto the $x_{j}$-axis
and $\top$ stands for transposition. The strain column%
\begin{equation}
\varepsilon=\left(  \varepsilon_{11},\varepsilon_{22},2^{1/2}\varepsilon
_{12},2^{1/2}\varepsilon_{13},2^{1/2}\varepsilon_{23},\varepsilon_{33}\right)
^{\top} \label{1.5}%
\end{equation}
substitutes for the strain tensor of rank $2$ with the Cartesian components
\begin{equation}
\varepsilon_{jk}\left(  u\right)  =\frac12\left(  \frac{\partial u_{j}%
}{\partial x_{k}}+\frac{\partial u_{k}}{\partial x_{j}}\right)  ,\qquad
j,k=1,2,3, \label{1.6}%
\end{equation}
and can be computed by the formula%
\begin{equation}
\varepsilon\left(  u\right)  =D\left(  \nabla\right)  u \label{1.00}%
\end{equation}
where $\nabla=\operatorname{grad}$ and
\begin{equation}
D\left(  \nabla\right)  ^{\top}=\left(
\begin{array}
[c]{cccccc}%
\partial_{1} & 0 & 2^{-1/2}\partial_{2} & 2^{-1/2}\partial_{3} & 0 & 0\\
0 & \partial_{2} & 2^{-1/2}\partial_{1} & 0 & 2^{-1/2}\partial_{3} & 0\\
0 & 0 & 0 & 2^{-1/2}\partial_{1} & 2^{-1/2}\partial_{2} & \partial_{3}%
\end{array}
\right)  ,\qquad\partial_{j}=\frac{\partial}{\partial x_{j}}. \label{1.7}%
\end{equation}
Notice that, according to the Mandel-Voigt notation, the factors $2^{1/2}$ and
$2^{-1/2}$ are introduced in (\ref{1.5}) and (\ref{1.7}) for the purpose of
equalizing the intrinsic norms of the tensor and the column of height $6$ (cf.
\cite{German, Nabook} and others).

The stress column
\[
\sigma=\left(  \sigma_{11},\sigma_{22},2^{1/2}\sigma_{12},2^{1/2}\sigma
_{13},2^{1/2}\sigma_{23},\sigma_{33}\right)  ^{\top}%
\]
of the same structure as in (\ref{1.5}) is to be found through the Hooke's law%
\begin{equation}
\sigma\left(  u\right)  =A\varepsilon\left(  u\right)  =AD(\nabla)u
\label{1.9}%
\end{equation}
where $A$ is the stiffness matrix of size $6\times6$, symmetric and positive
definite. This matrix contains elastic moduli, for instance, to an isotropic
elastic material there corresponds%
\begin{equation}
A=\left(
\begin{array}
[c]{cccccc}%
\lambda+2\mu & \lambda & 0 & 0 & 0 & \lambda\\
\lambda & \lambda+2\mu & 0 & 0 & 0 & \lambda\\
0 & 0 & 2\mu & 0 & 0 & 0\\
0 & 0 & 0 & 2\mu & 0 & 0\\
0 & 0 & 0 & 0 & 2\mu & 0\\
\lambda & \lambda & 0 & 0 & 0 & \lambda+2\mu
\end{array}
\right)  \label{1.10}%
\end{equation}
where $\lambda\geq0$ and $\mu>0$ are the Lam\'{e} constants.

In what follows we use matrix, rather than tensor, notation in elasticity
known as the Mandel-Voigt notation, cf. \cite{German}. In this way we write
the equilibrium equations as follows:%
\begin{equation}
L\left(  \nabla\right)  u\left(  h,x\right)  :=D\left(  -\nabla\right)
^{\top}AD\left(  \nabla\right)  u\left(  h,x\right)  =f\left(  h,x\right)
,\ \ \ x\in\Omega_{h}, \label{1.11}%
\end{equation}
where $f=\left(  f_{1},f_{2},f_{3}\right)  ^{\top}$ is the vector (column) of
the volume (mass) forces. The three-dimensional elasticity system (\ref{1.11})
is supplied with the traction-free boundary condition%
\begin{align}
N^{+}\left(  \nabla\right)  u\left(  h,x\right)   &  :=D\left(  e_{3}\right)
^{\top}AD\left(  \nabla\right)  u\left(  h,x\right)  =0,\ \ \ x\in\Sigma
_{h}^{+},\label{1.12}\\
N^{-}\left(  \nabla\right)  u\left(  h,x\right)   &  :=D\left(  -e_{3}\right)
^{\top}AD\left(  \nabla\right)  u\left(  h,x\right)  =0,\ \ \ x\in\Sigma
_{h}^{\bullet},\nonumber
\end{align}
where $e_{3}=\left(  0,0,1\right)  ^{\top}$ is the unit vector of the outward
normal on the bases $\Sigma_{h}^{\pm}$. At the clamped parts of the surfaces
$\partial\Omega_{h}$, we write
\begin{align}
u\left(  h,x\right)   &  =0,\ \ \ x\in\theta_{h},\label{1.13}\\
u\left(  h,x\right)   &  =0,\ \ \ x\in\upsilon_{h}. \label{1.14}%
\end{align}
We further refer to (\ref{1.12}) and (\ref{1.13}), (\ref{1.14}) as the Neumann
and Dirichlet conditions respectively.

The variational statement of problem (\ref{1.11})--(\ref{1.14}) reads: to find
a vector function $u\in H_{0}^{1}\left(  \Omega_{h};\Gamma_{h}\right)  ^{3}$
such that%
\begin{equation}
\left(  AD\left(  \nabla\right)  u,D\left(  \nabla\right)  v\right)
_{\Omega_{h}}=\left(  f,v\right)  _{\Omega_{h}},\ \ \ \forall v\in H_{0}%
^{1}\left(  \Omega_{h};\Gamma_{h}\right)  ^{3}, \label{1.15}%
\end{equation}
where $\left(  \ ,\ \right)  _{\Omega_{h}}$ is the natural scalar product in
the Lebesgue space $L^{2}\left(  \Omega_{h}\right)  $, $H_{0}^{1}\left(
\Omega_{h};\Gamma_{h}\right)  $ is a subspace of functions in the Sobolev
class $H^{1}\left(  \Omega_{h}\right)  $ which vanish at set (\ref{1.16}), and
the last superscript $3$ in (\ref{1.15}) means that test (vector) function $v$
has three components.

In view of the Dirichlet conditions (\ref{1.13}) and (\ref{1.14}) the Korn
inequality \cite{Korn}%
\begin{equation}
\left\Vert v;H^{1}\left(  \Omega_{h}\right)  \right\Vert \leq K_{h}\left\Vert
D\left(  \nabla\right)  v;L^{2}\left(  \Omega_{h}\right)  \right\Vert
,\ \ \ \forall v\in H_{0}^{1}\left(  \Omega_{h};\Gamma_{h}\right)  ^{3}
\label{1.17}%
\end{equation}
is valid and the unique solvability of the variational problem (\ref{1.15})
with any $f\in L^{2}\left(  \Omega_{h}\right)  ^{3}$ follows from the Riesz
representation theorem. Although there exist various approaches, see
\cite{DuLi, Fried, KoOl, Necas} and many others, to verify inequality
(\ref{1.17}), we still need to clarify the dependence of the Korn constant
$K_{h}$ on the small parameter $h$. To this end, we exhibit in Section
\ref{sect2} several variants of weighted anisotropic Korn's inequalities.

\subsection{Asymptotics and boundary layers\label{sect1.3}}

The Kirchhoff theory of thin plates created more than 150 years ago by means
of intuitive asymptotic analysis, has got a justification in miscellaneous
formulations by various methods and at different level of rigor, we refer only
to the mathematical monographs \cite{Ciarlet1, Ciarlet2, LeDret, Nabook, SaPa}
and intensive lists of literature therein, although setting aside a vast
volume of publications with important theoretical and practical results. If an
external loading is scaled to provide the elastic energy of a plate to gain
order $1=h^{0}$, the energy norm of asymptotic remainders in Kirchhoff's
asymptotic formulas becomes $O(\sqrt{h})$ and this is the best error estimate
available within a two-dimensional theory of plates. This limitation of the
asymptotic accuracy is prescribed exclusively by boundary layer effect near
the plate edge. Namely, in the vicinity of the lateral side $\upsilon_{h}$ the
standard plane-stress state in the $\left(  y_{1},y_{2}\right)  $-directions,
dominant in the midst of the plane, couples with a plane-strain state in the
$\left(  n,z\right)  $-directions, where $n$ is the normal vector to
$\partial\omega$. The latter is represented by special elastic fields which
slowly vary along the edge, produce strains and stresses of order 1 at the
lateral side $\upsilon_{h}$ but quickly, at the exponential rate, decay at a
distance from $\upsilon_{h}$. We again refer only to the mathematical papers
\cite{na251, na109} and \cite{DaugeR, DaugeG}, where the edge boundary layer
phenomenon was investigated in elasticity, and to Chapters 15 and 16 of the
monograph \cite{MaNaPl} where simplest scalar and general elliptic problems
were examined and basic principles to construct boundary layers are laid down.

The energy norm%
\begin{equation}
\left(  AD\left(  \nabla\right)  u,D\left(  \nabla\right)  v\right)
_{\Omega_{h}}^{1/2} \label{1.18}%
\end{equation}
of the boundary layer equals $O(\sqrt{h})$ that just predetermines the error
bound in an integral norm. However, the proximity of the plane-stress state in
a weighted H\"{o}lder norm is much less, cf. \cite{na224, na333}. For example,
bounds in the pointwise estimates of remainders in two-dimensional asymptotic
forms of stresses get the same order in $h$ as the main terms and this is
known in mechanics as ``edge effect'' in a thin plate. Thus, the error
$O(\sqrt{h})$ of the Kirchhoff model is mostly due to the intrinsic
localization of the edge effect in a $ch$-neighborhood of the lateral side
$\upsilon_{h}$.

The two-dimensional structure of boundary layers is kept for a smooth contour
$\partial\omega$ only. Plates with angulate edges were considered in
\cite{na145} where three-dimensional boundary layers with a power-law decay
rate were detected. However, their energy norm (\ref{1.18}) becomes
$o(\sqrt{h})$ so that they play a secondary role.

As it can be easily predicted due to the Sobolev embedding theorem, the
Dirichlet condition (\ref{1.13}) leads to the Sobolev condition (\ref{A6}) in
the limit, that is the overage deflection $w_{3}\in H^{2}(\omega)$ of the
plate vanishes at the point $\mathcal{O}$. However, the small support area
(\ref{1.3}) also provokes a fully three-dimensional boundary layer in the
vicinity of $\theta_{h}$ and provides a perturbation of order $|\ln h|^{-1/2}$
in the two-dimensional model. In other words, the convergence rate $O(|\ln
h|^{-1/2})$ of the rescaled three-dimensional displacements, cf. Theorem
\ref{th3.3}, to a solution of the limit Dirichlet-Sobolev problem (\ref{A7}),
(\ref{A8}), (\ref{A5}), (\ref{A6}) is very slow and, therefore,
unsatisfactory. At the same time, the asymptotic structures derived in this
paper provide the same accuracy as in the Kirchhoff theory.

In \cite{na472}, the same effect of a crucial reduction of the asymptotic
accuracy due to a small Dirichlet area has been observed and discussed in
detail for a scalar homogenization problem in a perforated planar domain.

\subsection{The elasticity capacity and self-adjoint extensions\label{sect1.4}%
}

To the best knowledge of the authors, this paper is the first mathematically
rigorous and complete study of the boundary layer effect near a small support
area and its influence on the whole stress-strain state of a plate. The
necessity to examine rapid variations of elastic fields in the vicinity of the
support area $\theta_{h}$ accounts for the following result obtained in the
consequential paper \cite{ButCaNa2}: the convergence rate to the
two-dimensional Kirchhoff solution of the true three-dimensional one is
extremely slow, of order $|\ln h|^{-1/2}$, but involving the boundary layer
reduces the error estimate bound down to $ch^{1/2}|\ln h|$ and makes it
acceptable for applications. At the same time, the whole boundary layer
solution is not known explicitly and can be determined by solving a
complicated elasticity problem (\ref{3.38})--(\ref{3.40}) in the spatial
infinite ply $\mathbb{R}^{2}\times(-1/2,1/2)$, for example, numerically. In
this way a simplification of asymptotic expansions and further modeling take
on special significance.

As usual in the theory of elliptic problems in singularly perturbed domains
(see the monographs \cite{Ilin, MaNaPl} and others) the alternation of
boundary conditions on a small set leads to asymptotic forms engaging singular
solutions of limit problems and in our case the Dirichlet condition
(\ref{1.13}) on $\theta_{h}$ affects the limit problem in $\omega$ in two
ways. First, it requires to impose the point condition $w_{3}(\mathcal{O})=0$,
(\ref{A6}), for the average deflection $w_{3}(y)$ of the plate. Second, the
far-field, eligible at a distance from the support $\theta_{h}$, contains the
Green matrix, cf. Section \ref{sect3.4} and see \cite{ButCaNa2}, whose entries
are solutions in $\omega$ with logarithmic singularities at the point $y=0$.
Due to the Sobolev embedding theorem $H^{2}\subset C$ in $\mathbb{R}^{2}$ the
point condition is well-posed within the Hilbert theory, cf. \cite[Ch.2]{LiMa}
but the Green matrix does not belong to the intrinsic energy space for the
Kirchhoff plate $\omega$ and therefore breaks the usual variational
formulation of the model.

To develop a two-dimensional model of the locally supported plate $\Omega_{h}%
$, we will turn in \cite{ButCaNa2} to the technique of self-adjoint extensions
of differential operators (see the pioneering paper \cite{BeFa} and, e.g., the
review \cite{Pav}). Namely, we get rid of the complicated three-dimensional
structure of the boundary layer and simulate its influence on the
stress-strain state of the plate by a proper choice of singular solutions. In
other words, we select a special self-adjoint extension in $L^{2}\left(
\omega\right)  ^{3}$ of the matrix $\mathcal{L}\left(  \nabla\right)  $ of
differential operators in the conventional Kirchhoff model described in
Section \ref{sect3.2}, formula (\ref{A13}). Parameters of this extension
depend only on the quantity $\left\vert \ln h\right\vert $ and the algebraic
characteristics $C^{\sharp}\left(  A,\Theta\right)  $ of the small support
area, namely the elastic logarithmic capacity matrix introduced heuristically
in \cite{na135}, a $4\times4$-matrix composed by coefficients in asymptotics
of solutions to the elasticity problem in $\Lambda$, see Section
\ref{sect3.5}, and is quite similar to the logarithmic capacity in harmonic
analysis, cf. \cite{Land, PoSe} and Remark \ref{remLOG}, and polarization and
capacity matrices in elasticity, cf. \cite{Ammari} and \cite{na280, na458}.
There exist numerical schemes \cite{na323, na404} to evaluate such characteristics.

\subsection{Architecture of the paper\label{sect1.5}}

In Section \ref{sect2} we derive anisotropic weighted Korn inequality in the
plate (\ref{1.1}) clamped along the lateral surface and the small set
(\ref{1.3}) (Theorem \ref{th2.1}) as well as demonstrate that weights
introduced in the Sobolev norms are optimal. Then we modify our approach to
prove Theorem \ref{th2.3} which serves for a plate with a traction-free
lateral surface but several small support areas.

The last section is devoted to three-dimensional boundary layers emerging in
the vicinity of small support areas. First of all, we outline the
two-dimensional Kirchhoff model of the plate $\Omega_{h}$ and the Sobolev
point condition $w_{3}(\mathcal{O})=0$, (\ref{A6}) for the average deflection
$w_{3}(y)$, which imitates the entire Dirichlet condition (\ref{1.13}) on
$\theta_{h}$. Moreover, we formulate Theorem \ref{th3.3} from \cite{ButCaNa2}
on the convergence of the spatial true solution to the Kirchhoff-Sobolev
solution in order to indicate the extremely slow convergence rate $O(|\ln
h|^{-1/2})$ and to emphasize the constitutive influence of boundary layers
described by an elasticity problem in the layer
\begin{equation}
\Lambda=\mathbb{R}^{2}\times\left(  -1/2,1/2\right)  \label{3.36}%
\end{equation}
clamped along the set $\theta$ on its planar base $\mathbb{R}^{2}%
\times\left\{  -1/2\right\}  $.

To prove the existence of a unique solution in the unbounded domain $\Lambda$
but with a finite energy, we need a new anisotropic weighted Korn inequality
which, however, can be derived by means of the same tricks as in Section
\ref{sect2}. At the same time, a rigorous derivation of asymptotic expansions
of the solution at infinity requires both, a formal procedure of dimension
reduction scheduled in Section \ref{sect3.2} and a specific application of the
Kondratiev theory \cite{Ko}, cf. \cite{na243}. This evident similarity of
several approaches in our paper explains its architecture.

In Section \ref{sect3.4} we strictly determine the elastic logarithmic
capacity matrix which was mentioned in Section \ref{sect1.4} and will be a
particular object in \cite{ButCaNa2}. In Theorem \ref{th3.7} we establish its
general properties.

\section{Weighted anisotropic inequalities of Korn's type\label{sect2}}

\subsection{Some Hardy-type inequalities}

Here below we collect some Hardy-type inequalities that will be used subsequently.

\begin{proposition}
Let $a(x)$ be a nonnegative function in $(0,1)$; then the following Hardy
inequality holds:
\[
\int_{0}^{1}a(x)|u(x)|^{2}\,dx\leq4\int_{0}^{1}\frac{1}{a(x)}\Big(\int_{x}%
^{1}a(t)\,dt\Big)^{2}|u^{\prime}(x)|^{2}\,dx\qquad\forall u\in H^{1}%
(0,1)\hbox{ with }u(0)=0.
\]

\end{proposition}

\begin{proof}
Since $u(0)=0$ we have
\[
u^{2}(x)=2\int_{0}^{x}u(t)u^{\prime}(t)\,dt,
\]
so that
\[
\int_{0}^{1}a(x)|u(x)|^{2}\,dx\leq2\int_{0}^{1}a(x)\Big(\int_{0}%
^{x}|u(t)||u^{\prime}(t)|\,dt\Big)\,dx.
\]
Interchanging the order of integration gives for the right-hand side
\[
2\int_{0}^{1}|u(t)||u^{\prime}(t)|\Big(\int_{t}^{1}a(x)\,dx\Big)\,dt,
\]
which, by H\"{o}lder inequality, is majorized by
\[
2\Big[\int_{0}^{1}a(t)|u(t)|^{2}\,dt\Big]^{1/2}\Big[\int_{0}^{1}\frac{1}%
{a(t)}\Big(\int_{t}^{1}a(x)\,dx\Big)^{2}|u^{\prime}(t)|^{2}\,dt\Big]^{1/2},
\]
and this concludes the proof.
\end{proof}

\bigskip

We list below some particular cases of functions $a(x)$, with the
corresponding Hardy-type inequalities. By taking
\[
a(x)=x^{-2}%
\]
we obtain for every $u\in H^{1}(0,1)$ with $u(0)=0$
\[
\int_{0}^{1}x^{-2}|u(x)|^{2}\,dx\leq4\int_{0}^{1}(1-x)^{2}|u^{\prime}%
(x)|^{2}\,dx,
\]
which implies the classical Hardy inequality
\begin{equation}
\int_{0}^{T}x^{-2}|u(x)|^{2}\,dx\leq4\int_{0}^{T}|u^{\prime}(x)|^{2}\,dx
\label{2.15}%
\end{equation}
for every $u\in H^{1}(0,T)$ with $u(0)=0$. Taking
\[
a(x)=(1-x)^{-1}|\ln(1-x)|^{-2}%
\]
gives for every $u\in H^{1}(0,1)$ with $u(0)=0$
\[
\int_{0}^{1}(1-x)^{-1}|\ln(1-x)|^{-2}|u(x)|^{2}\,dx\leq4\int_{0}%
^{1}(1-x)|u^{\prime}(x)|^{2}\,dx,
\]
which implies the inequality
\begin{equation}
\int_{0}^{R}x^{-1}|\ln(x/R)|^{-2}|u(x)|^{2}\,dx\leq4\int_{0}^{R}x|u^{\prime
}(x)|^{2}\,dx \label{2.16}%
\end{equation}
for every $u\in H^{1}(0,R)$ with $u(R)=0$. Taking
\[
a(x)=x^{-3}|\ln(x/2)|^{-2}%
\]
gives for every $u\in H^{1}(0,1)$ with $u(0)=0$
\[
\int_{0}^{1}x^{-3}|\ln(x/2)|^{-2}|u(x)|^{2}\,dx\leq4\int_{0}^{1}x^{-1}%
|\ln(x/2)|^{-2}|u^{\prime}(x)|^{2}\,dx
\]
which implies the inequality
\begin{equation}
\int_{0}^{R/2}x^{-3}|\ln(x/R)|^{-2}|u(x)|^{2}\,dx\leq4\int_{0}^{R/2}x^{-1}%
|\ln(x/R)|^{-2}|u^{\prime}(x)|^{2}\,dx \label{2.21}%
\end{equation}
for every $u\in H^{1}(0,R/2)$ with $u(0)=0$. Taking
\[
a(x)=(x+h)^{-4}\qquad\hbox{with }h>0
\]
gives for every $u\in H^{1}(0,1)$ with $u(0)=0$
\[
\int_{0}^{1}(x+h)^{-4}|u(x)|^{2}\,dx\leq\frac{4}{9}\int_{0}^{1}(x+h)^{-2}%
|u^{\prime}(x)|^{2}\,dx,
\]
which implies the inequality
\begin{equation}
\int_{0}^{T}(x+h)^{-4}|u(x)|^{2}\,dx\leq\frac{4}{9}\int_{0}^{T}(x+h)^{-2}%
|u^{\prime}(x)|^{2}\,dx \label{2.22}%
\end{equation}
for every $u\in H^{1}(0,T)$ with $u(0)=0$.

\subsection{Anisotropic weighted Korn inequality\label{sect2.2}}

In the paper \cite{Shoikhet} devoted to justification of the Kirchhoff theory
of plates, see also the pioneering papers \cite{Morgen} and \cite{Ciarle}
together with the monographs \cite{Ciarlet1, Ciarlet2, LeDret, Nabook, SaPa}
etc., it was proved that a constant $K\left(  \omega\right)  $ in the
inequality%
\begin{equation}
|||u;\Omega_{h}|||\leq K\left(  \omega\right)  \left\Vert D\left(
\nabla\right)  u;L^{2}\left(  \Omega_{h}\right)  \right\Vert ,\ \ \ \ \forall
u\in H_{0}^{1}\left(  \Omega_{h},\upsilon_{h}\right)  \label{2.1}%
\end{equation}
does not depend on $h\in(0,1]$ and, of course, on $u$, where $|||u;\Omega
_{h}|||$ is the anisotropic Sobolev norm
\begin{align}
|||u;\Omega_{h}|||^{2}  &  =\int_{\Omega_{h}}\Bigg[\sum_{i=1}^{2}\left(
|\nabla_{y}u_{i}|^{2}+h^{2}\left(  \left\vert \frac{\partial u_{i}}{\partial
z}\right\vert ^{2}+\left\vert \frac{\partial u_{3}}{\partial y_{i}}\right\vert
^{2}\right)  +|u_{i}|^{2}\right) \label{2.2}\\
&  \ \ \ \ \ \ \ \ \ \ \ \ \ \ \ \ \ \ \ \ \ +\left\vert \partial_{z}%
u_{3}\right\vert ^{2}+h^{2}|u_{3}|^{2}\Bigg]\,dx.\nonumber
\end{align}
Notice that the distribution of coefficients $h^{p}$ in (\ref{2.2}) is
optimal, namely one cannot replace $h^{2}$ by $h^{\alpha}$, with $\alpha<2$,
without losing the independence property of $K(\omega)$.

In the paper \cite{na155} a convenient modification of norm (\ref{2.2}) was
suggested, namely the weighted anisotropic norm
\begin{align}
|||u;\Omega_{h}|||_{0}^{2}  &  =\int_{\Omega_{h}}\Bigg[\sum_{i=1}^{2}\left(
\left\vert \nabla_{y}u_{i}\right\vert ^{2}+\frac{h^{2}}{s_{h}^{2}}\left(
\left\vert \frac{\partial u_{i}}{\partial z}\right\vert ^{2}+\left\vert
\frac{\partial u_{3}}{\partial y_{i}}\right\vert ^{2}\right)  +\frac{1}%
{s_{h}^{2}}|u_{i}|^{2}\right) \label{2.3}\\
&  \ \ \ \ \ \ \ \ \ \ \ \ \ \ \ \ \ \ \ \ \ \ +\left\vert \partial_{z}%
u_{3}\right\vert ^{2}+\frac{h^{2}}{s_{h}^{4}}|u_{3}|^{2}\Bigg]dx\nonumber
\end{align}
with the weighting function%
\begin{equation}
s_{h}\left(  y\right)  =h+\text{dist}\left(  y,\partial\omega\right)  .
\label{2.4}%
\end{equation}
Notice that in the middle of the plate the weights in (\ref{2.2}) and
(\ref{2.3}) are equivalent but in the vicinity of the lateral side
$\upsilon_{h}$ where the Dirichlet condition (\ref{1.14}) is imposed, the
displacements $u_{i}$, $u_{3}$ and all their derivatives get multipliers
$h^{-1}$ and $h^{0}=1$ respectively.

Our further justification scheme requires to insert into the norm (\ref{2.3})
some additional weights
\begin{equation}
S_{hq}(y)=(h^{2}+|y|^{2})^{-q/2}\left(  1+\left\vert \ln(h^{2}+|y|^{2}%
)\right\vert \right)  ^{-1} \label{2.5}%
\end{equation}
that become big in a $ch$-neighborhood of the set $\overline{\theta_{h}}$ and,
therefore, take the Dirichlet condition (\ref{1.13}) into account. Note that
the function $S_{hk}$ is smooth with any $h>0$ and in the sequel it is proper
to put into (\ref{2.3}) the weight $h+s_{0}\left(  y\right)  $, tantamount to
(\ref{2.4}). Here, $s_{0}$ stands for a smooth in $\overline{\omega}$ and
positive in $\omega$ function equivalent to dist$\left(  \cdot,\partial
\omega\right)  $ in the vicinity of the boundary $\partial\omega$.

Then we will prove the anisotropic weighted Korn inequality in the following theorem.

\begin{theorem}
\label{th2.1}Any displacement field $u\in H^{1}\left(  \Omega_{h}\right)
^{3}$ satisfying the Dirichlet conditions (\ref{1.13}) and (\ref{1.14}), meets
the weighted anisotropic Korn inequality
\begin{equation}
|||u;\Omega_{h}|||_{\bullet}\leq K_{\bullet}\left(  \omega\right)  \left\Vert
D\left(  \nabla\right)  u;L^{2}\left(  \Omega_{h}\right)  \right\Vert
,\ \ \ \ \forall u\in H_{0}^{1}\left(  \Omega_{h},\upsilon_{h}\cup\theta
_{h}\right)  \label{2.6}%
\end{equation}
where%
\begin{align}
|||u;\Omega_{h}|||_{\bullet}^{2}  &  =\int_{\Omega_{h}}\Bigg[\sum_{i=1}%
^{2}\left(  \left\vert \nabla_{y}u_{i}\right\vert ^{2}+\frac{h^{2}}{s_{h}^{2}%
}S_{h1}^{2}\left(  \left\vert \frac{\partial u_{i}}{\partial z}\right\vert
^{2}+\left\vert \frac{\partial u_{3}}{\partial y_{i}}\right\vert ^{2}\right)
+\frac{1}{s_{h}^{2}}S_{h1}^{2}|u_{i}|^{2}\right) \label{2.7}\\
&  \ \ \ \ \ \ \ \ \ \ \ \ \ \ \ \ \ \ \ \ \ +\left\vert \partial_{z}%
u_{3}\right\vert ^{2}+\frac{h^{2}}{s_{h}^{4}}S_{h2}^{2}|u_{3}|^{2}%
\Bigg]\,dx.\nonumber
\end{align}
with a constant $K\left(  \omega\right)  $ independent of $h\in\left(
0,h_{0}\right]  $ and $u\in H_{0}^{1}\left(  \Omega_{h};\upsilon_{h}\cup
\theta_{h}\right)  ^{3}$.
\end{theorem}

\begin{proof}
Due to a completion argument, it suffices to verify inequality (\ref{2.6}) for
any smooth function $u$ in $\overline{\Omega_{h}}$ which vanishes near
$\overline{\upsilon_{h}}$ and $\overline{\theta_{h}}$, see (\ref{1.2}). We use
the stretched coordinates%
\begin{equation}
\xi=\left(  \eta,\zeta\right)  =\left(  h^{-1}y,h^{-1}z\right)  \label{3.35}%
\end{equation}
and consider the circular cylinder $\mathbf{Q}_{hR}=\mathbb{B}_{hR}^{2}%
\times\left(  -h/2,h/2\right)  $ where $\mathbb{B}_{\rho}^{2}$ is the disk
$\left\{  y:|y|<\rho\right\}  $ and the radius $R$ is fixed such that
$\theta\subset\mathbb{B}_{R/2}^{2}$. We write
\begin{equation}%
\begin{split}
\left\Vert \nabla_{x}u;L^{2}\left(  \mathbf{Q}_{hR}\right)  \right\Vert ^{2}
&  +h^{-2}\left\Vert u;L^{2}\left(  \mathbf{Q}_{hR}\right)  \right\Vert
^{2}=h\left\Vert \nabla_{\xi}u;L^{2}\left(  \mathbf{Q}_{R}\right)  \right\Vert
^{2}+h^{-1}\left\Vert u;L^{2}\left(  \mathbf{Q}_{R}\right)  \right\Vert ^{2}\\
&  =h\left\Vert \xi\mapsto u(x);H^{1}\left(  \mathbf{Q}_{R}\right)
\right\Vert ^{2}\leq c(R,\omega)h\left\Vert D(\nabla_{\xi})u;L^{2}\left(
\mathbf{Q}_{R}\right)  \right\Vert ^{2}\\
&  =c(R,\omega)\left\Vert D(\nabla_{x})u;L^{2}\left(  \mathbf{Q}_{hR}\right)
\right\Vert ^{2}\leq c(R,\omega)\left\Vert D(\nabla)u;L^{2}(\Omega
_{h})\right\Vert ^{2}.
\end{split}
\label{2.9}%
\end{equation}
Here, we have applied the standard Korn inequality%
\[
\left\Vert v;H^{1}\left(  \mathbf{Q}_{R}\right)  \right\Vert ^{2}\leq c\left(
R,\omega\right)  \left\Vert D\left(  \nabla_{\xi}\right)  v;L^{2}\left(
\mathbf{Q}_{R}\right)  \right\Vert ^{2},\ \ \ \forall v\in H_{0}^{1}\left(
\mathbf{Q}_{R};\theta\right)  ^{3}%
\]
which holds true due to the Dirichlet condition on $\theta$.

Let $\chi\in C^{\infty}(\mathbb{R})$ be a reference cut-off function such
that
\begin{equation}
\chi\left(  r\right)  =1\text{\ for\ }r<1/2\text{ \ \ and \ \ }\chi\left(
r\right)  =0\text{ for\ }r\geq1,\ \ \ 0\leq\chi\leq1. \label{2.11}%
\end{equation}
We set $X_{h}(y)=1-\chi\left(  (hR)^{-1}|y|\right)  $ and, in view of
$\left\vert \nabla_{y}X_{h}\left(  y\right)  \right\vert \leq ch^{-1}$,
obtain
\begin{equation}%
\begin{split}
&  \left\Vert D\left(  \nabla\right)  \left(  X_{h}u\right)  ;L^{2}\left(
\Omega_{h}\right)  \right\Vert ^{2}\leq c(\left\Vert D(\nabla)u;L^{2}%
(\Omega_{h})\right\Vert ^{2}\\
&  \qquad+h^{-2}\left\Vert u;L^{2}\left(  \left(  \mathbb{B}_{hR}%
\setminus\mathbb{B}_{hR/2}\right)  \times\left(  -h/2,h/2\right)  \right)
\right\Vert ^{2})\leq c\left\Vert D(\nabla)u;L^{2}(\Omega_{h})\right\Vert
^{2}.
\end{split}
\label{2.12}%
\end{equation}
Since $S_{h1}(y)\leq ch^{-1}$ in $\mathbf{Q}_{hR}$, estimates (\ref{2.9}) and
(\ref{2.12}) mean that we further can treat the product $X_{h}u$ only. This
product is still denoted by $u$ but we remember its specific property
\begin{equation}
u(y,z)=0\quad\text{for }|y|<hR/2. \label{2.0}%
\end{equation}
Formulas (\ref{1.6}), (\ref{1.14}) and integrating by parts yield
\[%
\begin{split}
4\int_{\Omega_{h}}\left\vert \varepsilon_{12}(u)\right\vert ^{2}dx  &
=\int_{\Omega_{h}}\left(  \left\vert \frac{\partial u_{1}}{\partial y_{2}%
}\right\vert ^{2}+\left\vert \frac{\partial u_{2}}{\partial y_{1}}\right\vert
^{2}+2\frac{\partial u_{1}}{\partial y_{2}}\frac{\partial u_{2}}{\partial
y_{1}}\right)  dx\\
&  =\int_{\Omega_{h}}\left(  \left\vert \frac{\partial u_{1}}{\partial y_{2}%
}\right\vert ^{2}+\left\vert \frac{\partial u_{2}}{\partial y_{1}}\right\vert
^{2}+2\frac{\partial u_{1}}{\partial y_{2}}\frac{\partial u_{2}}{\partial
y_{1}}\right)  dx\\
&  \geq\int_{\Omega_{h}}\left(  \left\vert \frac{\partial u_{1}}{\partial
y_{2}}\right\vert ^{2}+\left\vert \frac{\partial u_{2}}{\partial y_{1}%
}\right\vert ^{2}-\left\vert \frac{\partial u_{1}}{\partial y_{1}}\right\vert
^{2}-\left\vert \frac{\partial u_{2}}{\partial y_{2}}\right\vert ^{2}\right)
dx.
\end{split}
\]
Hence, recalling the strains $\varepsilon_{11}\left(  u\right)  $ and
$\varepsilon_{22}\left(  u\right)  $, we obtain%
\begin{equation}
\left\Vert \nabla_{y}u_{i};L^{2}\left(  \Omega_{h}\right)  \right\Vert
^{2}\leq2\left\Vert D\left(  \nabla\right)  u;L^{2}\left(  \Omega_{h}\right)
\right\Vert ^{2},\ \ \ i=1,2. \label{2.13}%
\end{equation}
We have the Friedrichs inequality in $\omega$ integrated in $z\times
(-h/2,h/2)$, that is
\begin{equation}
\left\Vert u_{i};L^{2}\left(  \Omega_{h}\right)  \right\Vert ^{2}\leq
c_{\omega}\left\Vert \nabla_{y}u_{i};L^{2}\left(  \Omega_{h}\right)
\right\Vert ^{2}. \label{2.14}%
\end{equation}
Let us explain how inequalities (\ref{2.13})--(\ref{2.14}) provide the
estimates%
\begin{equation}
\left\Vert s_{h}^{-1}S_{01}u_{i};L^{2}\left(  \Omega_{h}\right)  \right\Vert
^{2}\leq c\left\Vert D\left(  \nabla\right)  u;L^{2}\left(  \Omega_{h}\right)
\right\Vert ^{2},\qquad i=1,2, \label{2.17}%
\end{equation}
where $S_{01}\left(  y\right)  $ is given by (\ref{2.5}) at $h=0$. To attach
the weight $s_{h}^{-1}$, we rewrite the function $u_{i}$ in the natural
curvilinear coordinates $n,s$ while $n$ is the oriented distance to
$\partial\omega$, $n<0$ in $\omega$, and $s$ is the arc length along
$\partial\omega$. Setting $t=h-n$ and $U\left(  t\right)  =u_{i}\left(
n,s\right)  $ in (\ref{2.15}), we observe that the Jacobian of the coordinate
change $y\mapsto\left(  n,s\right)  $ and its inverse are bounded in the
$\delta$-neighborhood $\mathcal{V}_{\delta}$ of $\partial\omega$, $\delta>0$
being fixed small. Since $\left\vert \partial_{|n|}u_{i}(x)\right\vert
\leq\left\vert \nabla_{y}u_{i}(x)\right\vert $, integrating in $s\in\omega$
and $z\in\left(  -h/2,h/2\right)  $ converts (\ref{2.15}) into the inequality%
\begin{equation}
\left\Vert s_{h}^{-1}u_{i};L^{2}\left(  \left(  \omega\cap\mathcal{V}_{\delta
}\right)  \times\left(  -h/2,h/2\right)  \right)  \right\Vert ^{2}\leq
c\left\Vert \nabla_{y}u_{i};L^{2}\left(  \left(  \omega\cap\mathcal{V}%
_{\delta}\right)  \times\left(  -h/2,h/2\right)  \right)  \right\Vert ^{2}
\label{2.18}%
\end{equation}
because $s_{h}(y)\sim h+|n|$. Formula (\ref{2.14}) allows us to replace
$\omega\cap\mathcal{V}_{\delta}$ with $\omega$ in (\ref{2.18}).

The weight $S_{01}$ ought to be inserted in a similar manner. However, we go
over to the polar coordinate system $\left(  r,\varphi\right)  $ in the
$y$-plane and multiply $u_{i}$ with the cut-off function $\chi_{\omega}\left(
y\right)  =\chi\left(  y/R_{\omega}\right)  $ where $R_{\omega}$ is chosen
such that $\mathbb{B}_{R_{\omega}}^{2}\subset\omega$. It remains to mention
that $dy=rdrd\varphi$ and%
\[
\left\Vert \partial_{r}\left(  \chi_{\omega}u_{i}\right)  ;L^{2}\left(
\mathbb{B}_{R_{\omega}}^{2}\right)  \right\Vert \leq c_{\omega}\left\Vert
u_{i};H^{1}\left(  \mathbb{B}_{R_{\omega}}^{2}\right)  \right\Vert ,
\]
and to apply inequalities (\ref{2.16}) integrated over $\left(  \varphi
,z\right)  \in\left(  0,2\pi\right)  \times\left(  -h/2,h/2\right)  $ and
(\ref{2.14}), (\ref{2.13}).

To get an accurate information on the derivatives $\partial_{z}u_{i}$ and
$\nabla_{y}u_{3}$ is a much more complicated issue and we apply a trick from
\cite{na155}. To this end, we introduce the cut-off function $\chi_{h}\left(
z\right)  =\chi\left(  2z/h\right)  $ which is null on the bases $\Sigma
_{h}^{\pm}$ and observe that, according to definition (\ref{2.4}), (\ref{2.5})
of the weights, there holds%
\begin{equation}
0<h^{2}s_{h}\left(  y\right)  ^{-1}S_{01}\left(  y\right)  \leq c_{S}%
,\ \ \ y\in\overline{\omega}\setminus\mathbb{B}_{hR/2}^{2},\ h\in\left(
0,h_{0}\right]  . \label{2.19}%
\end{equation}
Notice that we have reduced our analysis to the case $u=0$ in $\mathbb{Q}%
_{hR/2}=\mathbb{B}_{hR/2}^{2}\times(-h/2,h/2)$. Then we proceed as follows:
\begin{gather*}
c_{S}^{2}\int_{\Omega_{h}}\left\vert \varepsilon_{i3}(u)\right\vert ^{2}dx\geq
h^{2}\int_{\Omega_{h}}\chi_{h}^{2}s_{h}^{-2}S_{01}^{2}\left(  \left\vert
\frac{\partial u_{i}}{\partial z}\right\vert ^{2}+\left\vert \frac{\partial
u_{3}}{\partial y_{i}}\right\vert ^{2}\right)  dx\\
+2h^{2}\int_{\Omega_{h}}\chi_{h}^{2}s_{h}^{-2}S_{01}^{2}\frac{\partial u_{i}%
}{\partial z}\frac{\partial u_{3}}{\partial y_{i}}dx=:I_{1}^{h}+2I_{2}^{h},\\
I_{2}^{h}=-h^{2}\int_{\Omega_{h}}\chi_{h}^{2}s_{h}^{-2}S_{01}^{2}u_{i}%
\frac{\partial^{2}u_{3}}{\partial z\partial y_{i}}dx-2h^{2}\int_{\Omega_{h}%
}\chi_{h}\partial_{z}\chi_{h}s_{h}^{-2}S_{01}^{2}u_{i}\frac{\partial u_{3}%
}{\partial y_{i}}dx=:I_{3}^{h}+2I_{4}^{h}.
\end{gather*}
Since $\left\vert \partial_{z}\chi_{h}\left(  z\right)  \right\vert \leq
c_{\chi}h^{-1}$, we have%
\begin{align*}
\left\vert I_{4}^{h}\right\vert  &  \leq c_{\chi}\left(  h^{2}\int_{\Omega
_{h}}\chi_{h}^{2}s_{h}^{-2}S_{01}^{2}\left\vert \frac{\partial u_{3}}{\partial
y_{i}}\right\vert ^{2}dx\right)  ^{1/2}\left(  \int_{\Omega_{h}}s_{h}%
^{-2}S_{01}^{2}\left\vert u_{i}\right\vert ^{2}dx\right)  ^{1/2}\\
&  \leq\delta I_{1}^{h}+C\delta^{-}1\left\Vert s_{h}^{-}1S_{0}1u_{i}%
;L^{2}\left(  \Omega_{h}\right)  \right\Vert ^{2}%
\end{align*}
where $\delta>0$ is arbitrary and the last norm has been estimated in
(\ref{2.17}). Furthermore,
\[
I_{3}^{h}=h^{2}\int_{\Omega_{h}}\chi_{h}^{2}s_{h}^{-2}S_{01}^{2}\frac{\partial
u_{i}}{\partial y_{i}}\frac{\partial u_{3}}{\partial z}dx+h^{2}\int
_{\Omega_{h}}\chi_{h}^{2}u_{i}\frac{\partial u_{3}}{\partial z}\frac{\partial
}{\partial y_{i}}\left(  s_{h}^{-2}S_{01}^{2}\right)  dx=:I_{5}^{h}+I_{6}%
^{h},
\]
both the integrals getting appropriate bounds with simplicity. First, in view
of (\ref{2.19}),
\[
\left\vert I_{5}^{h}\right\vert \leq\frac{1}{2}c_{S}^{2}\left(  \left\Vert
\varepsilon_{ii}\left(  u\right)  ;L^{2}\left(  \Omega_{h}\right)  \right\Vert
^{2}+\left\Vert \varepsilon_{33}\left(  u\right)  ;L^{2}\left(  \Omega
_{h}\right)  \right\Vert ^{2}\right)  \leq C\left\Vert D\left(  \nabla\right)
u;L^{2}\left(  \Omega_{h}\right)  \right\Vert ^{2}.
\]
Second, performing differentiation of weight in $y$ we see that, for
$y\in\overline{\omega}\setminus\mathbb{B}_{hR/2}^{2}$ and $h\in\left(
0,h_{0}\right]  $%
\[
h^{2}\left\vert \frac{\partial}{\partial y_{i}}\left(  s_{h}\left(  y\right)
^{-2}S_{01}\left(  y\right)  ^{2}\right)  \right\vert \leq ch^{2}s_{h}\left(
y\right)  ^{-2}S_{01}\left(  y\right)  ^{2}\left(  \frac{1}{h+s_{0}\left(
y\right)  }+\frac{1}{\left\vert y\right\vert }\right)  \leq Cs_{h}\left(
y\right)  ^{-1}S_{01}\left(  y\right)  ,
\]
and, thus,%
\[
\left\vert I_{6}^{h}\right\vert \leq c\left(  \left\Vert s_{h}^{-1}S_{01}%
u_{i};L^{2}\left(  \Omega_{h}\right)  \right\Vert ^{2}+\left\Vert \partial
_{z}u_{3};L^{2}\left(  \Omega_{h}\right)  \right\Vert ^{2}\right)  \leq
C\left\Vert D\left(  \nabla\right)  u;L^{2}\left(  \Omega_{h}\right)
\right\Vert ^{2}.
\]

Collecting relations listed above provides the formula%
\[
I_{1}^{h}\leq c\left(  1+\delta^{-1}\right)  \left\Vert D\left(
\nabla\right)  u;L^{2}\left(  \Omega_{h}\right)  \right\Vert ^{2}-2\delta
I_{1}^{h}%
\]
so that putting $\delta=1/4$ leads to the following estimate of the
derivatives in question:
\begin{equation}
\label{2.20}%
\begin{split}
h^{2}\left\Vert s_{h}^{-1}S_{01}\partial_{z} u_{i};L^{2}\left(  \Omega
_{h/2}\right)  \right\Vert ^{2}  &  +h^{2}\left\Vert s_{h}^{-1}S_{01}\partial
u_{3}/\partial y_{i};L^{2}\left(  \Omega_{h/2}\right)  \right\Vert ^{2}\\
&  \le C\left\Vert D(\nabla)u;L^{2}\left(  \Omega_{h}\right)  \right\Vert
^{2},
\end{split}
\end{equation}
however only inside the thinner plate $\Omega_{h/2}=\omega\times\left(
-h/4,h/4\right)  $ where the cut-off function equals one.

We postpone spreading of estimate (\ref{2.20}) onto the entire plate and
conclude with the displacement $u_{3}$ itself.

Operating with (\ref{2.22}) and (\ref{2.21}) in the same way as with
(\ref{2.15}) and (\ref{2.16}), we derive from estimates (\ref{2.20}) of
$\nabla_{y} u_{3}$ the formula
\begin{equation}
h^{2}\left\Vert s_{h}^{-1}S_{02}u_{3};L^{2}(\Omega_{h/2}) \right\Vert ^{2}\le
c\left\Vert D(\nabla)u;L^{2}(\Omega_{h})\right\Vert ^{2} \label{2.23}%
\end{equation}
with a weight required in (\ref{2.7}) but again in a thinner plate. By the
way, we become in position to lighten weights in (\ref{2.20}), (\ref{2.23}) by
the replacement $S_{0q}\mapsto S_{hq}$ and write
\begin{equation}
|||u;\Omega_{h/2}|||_{\bullet}^{2}\le c\left\Vert D(\nabla)u;L^{2}(\Omega
_{h})\right\Vert ^{2}. \label{2.000}%
\end{equation}
Moreover, we now may forget about the artificial property (\ref{2.0}) of the
displacement field $u$.

To improve the obtained estimates, we apply a method of passing anisotropic
Korn inequalities from one part of a body to another part which was proposed
in \cite{na155} and elaborated in \cite{na398}. Let $\mathbb{Q}_{h}$ be the
cube $\left\{  x:\left\vert y_{i}-y_{i}^{0}\right\vert <h/2,\ |z|<h/2\right\}
$ with some center $(y^{0},0)$. Extending $u$ by zero from $\Omega_{h}$ onto
the layer $\left\{  x:|z|<h/2\right\}  $, we assume that $y^{0}\in\omega$ and
set $U(\eta,\zeta)=u(y^{0}+h\eta,h\zeta)$ where $\xi=(\eta,\zeta)=\left(
h^{-1}\left(  y-y^{0}\right)  ,h^{-1}z\right)  \in\mathbb{Q}_{1}$ are
stretched coordinates, cf. (\ref{3.35}). We introduce the rigid motion matrix
of size $3\times6$
\begin{equation}
d\left(  \xi\right)  =\left(
\begin{array}
[c]{cccccc}%
1 & 0 & 0 & 0 & \xi_{3} & -\xi_{2}\\
0 & 1 & 0 & -\xi_{3} & 0 & \xi_{1}\\
0 & 0 & 1 & \xi_{2} & -\xi_{1} & 0
\end{array}
\right)  \label{2.24}%
\end{equation}
and make the following decomposition in the unit cube $\mathbb{Q}_{1}$:%
\begin{equation}
\mathcal{U}\left(  \xi\right)  =\mathcal{U}^{\bot}\left(  \xi\right)
+d\left(  \xi\right)  \mathcal{U}^{0}. \label{2.25}%
\end{equation}
Here, the last term implies a rigid motion generated by the column%
\begin{equation}
\mathcal{U}^{0}=\mathbf{d}\left(  \mathbb{Q}_{1}^{\prime}\right)  ^{-1}%
\int_{\mathbb{Q}_{1}^{\prime}}d\left(  \xi\right)  ^{\top}\mathcal{U}\left(
\xi\right)  d\xi\in\mathbb{R}^{6}, \label{2.26}%
\end{equation}
where $\mathbb{Q}_{1}^{\prime}=\left\{  \xi\in\mathbb{Q}_{1}:|\zeta
|<1/4\right\}  $ is a half of the cube and $\mathbf{d}\left(  \mathbb{Q}%
_{1}^{\prime}\right)  $ is a Gram matrix of size $6\times6$, symmetric and
positive definite,%
\begin{equation}
\mathbf{d}\left(  \mathbb{Q}_{1}^{\prime}\right)  =\int_{\mathbb{Q}%
_{1}^{\prime}}d\left(  \xi\right)  ^{\top}d\left(  \xi\right)  d\xi.
\label{2.27}%
\end{equation}
Owing to definition (\ref{2.26}), (\ref{2.27}), the component $\mathcal{U}%
^{\bot}$ meets the orthogonality conditions%
\begin{equation}
\int_{\mathbb{Q}_{1}^{\prime}}d\left(  \xi\right)  ^{\top}\mathcal{U}^{\bot
}\left(  \xi\right)  d\xi=0\in\mathbb{R}^{6} \label{2.28}%
\end{equation}
which, as known (see the proof of Theorem 3.3.3 \cite{DuLi} or Theorem 2.3.3
\cite{Nabook}), assure the Korn inequality on the intact cube%
\begin{equation}
\left\Vert \mathcal{U}^{\bot};H^{1}\left(  \mathbb{Q}_{1}\right)  \right\Vert
^{2}\leq K\left\Vert D\left(  \nabla_{\xi}\right)  \mathcal{U}^{\bot}%
;L^{2}\left(  \mathbb{Q}_{1}\right)  \right\Vert ^{2}. \label{2.29}%
\end{equation}
Due to the central symmetry of the integration domain in (\ref{2.27}) the
matrix $\mathbf{d}\left(  \mathbb{Q}_{1}^{\prime}\right)  $ is diagonal.
Hence, formulas (\ref{2.26}) and (\ref{2.24}) immediately show that%
\begin{equation}
\left\vert \mathcal{U}_{i}^{0}\right\vert \leq c\left\Vert \mathcal{U}%
_{i};L^{2}\left(  \mathbb{Q}_{1}^{\prime}\right)  \right\Vert
,\ \ \ i=1,2,\ \ \ \ \ \left\vert \mathcal{U}_{3}^{0}\right\vert \leq
c\left\Vert \mathcal{U}_{3};L^{2}\left(  \mathbb{Q}_{1}^{\prime}\right)
\right\Vert . \label{2.30}%
\end{equation}
The component $\mathcal{U}_{6}^{0}$ is estimated in the following way:%
\[%
\begin{split}
\left\vert \mathcal{U}_{6}^{0}\right\vert  &  =\mathbf{d}_{66}^{-1}\left\vert
\int_{\mathbb{Q}_{1}^{\prime}}\left(  \eta_{1}\mathcal{U}_{2}\left(
\xi\right)  -\eta_{2}\mathcal{U}_{1}\left(  \xi\right)  \right)
d\xi\right\vert \\
&  =\mathbf{d}_{66}^{-1}\left\vert \int_{\mathbb{Q}_{1}^{\prime}}\left(
\mathcal{U}_{2}\left(  \xi\right)  \frac{\partial}{\partial\eta_{1}}\left(
\frac{\eta_{1}^{2}}2-\frac1{8}\right)  -\mathcal{U}_{1}\left(  \xi\right)
\frac{\partial}{\partial\eta_{2}}\left(  \frac{\eta_{2}^{2}}2-\frac
1{8}\right)  \right)  d\xi\right\vert \\
&  =\mathbf{d}_{66}^{-1}\left\vert \int_{\mathbb{Q}_{1}^{\prime}}\left(
\left(  \frac{\eta_{1}^{2}}2-\frac1{8}\right)  \frac{\partial\mathcal{U}_{2}%
}{\partial\eta_{1}}\left(  \xi\right)  -\left(  \frac{\eta_{2}^{2}}2-\frac
1{8}\right)  \frac{\partial\mathcal{U}_{1}}{\partial\eta_{2}}\left(
\xi\right)  \right)  d\xi\right\vert \\
&  \leq c\left(  \left\Vert \frac{\partial\mathcal{U}_{1}}{\partial\eta_{2}%
};L^{2}(\mathbb{Q}_{1}^{\prime})\right\Vert +\left\Vert \frac{\partial
\mathcal{U}_{2}}{\partial\eta_{1}};L^{2}\left(  \mathbb{Q}_{1}^{\prime
}\right)  \right\Vert \right)  .
\end{split}
\]
Note that integration by parts did not bring a surface integral because
$\left(  \eta_{i}^{2}/2-1/8\right)  _{\eta_{i}=\pm1/2}=0$. Referring to the
formulas%
\[
\zeta=\frac{\partial}{\partial\zeta}\left(  \frac{\zeta^{2}}2-\frac
1{32}\right)  ,\qquad\left.  \left(  \frac{\zeta^{2}}2-\frac1{32}\right)
\right\vert _{\zeta=\pm1/4}=0,
\]
we finally obtain in a similar manner that
\[
\left\vert \mathcal{U}_{6-i}^{0}\right\vert \le c\left(  \left\Vert
\frac{\partial\mathcal{U}_{i}}{\partial\zeta};L^{2}\left(  \mathbb{Q}%
_{1}^{\prime}\right)  \right\Vert +\left\Vert \frac{\partial\mathcal{U}_{3}%
}{\partial\eta_{i}};L^{2}\left(  \mathbb{Q}_{1}^{\prime}\right)  \right\Vert
\right)  ,\qquad i=1,2.
\]

Now we return to the $x$-coordinates and the displacement field $u$. The
decomposition (\ref{2.25}) determines the component $u^{\bot}$ which,
according to (\ref{2.29}), gets the estimate, cf. (\ref{2.9}),%
\begin{equation}
\left\Vert \nabla u^{\bot};L^{2}\left(  \mathbb{Q}_{h}\right)  \right\Vert
^{2}+h^{-2}\left\Vert u^{\bot};L^{2}\left(  \mathbb{Q}_{h}\right)  \right\Vert
^{2}\leq c\left\Vert D\left(  \nabla\right)  u;L^{2}\left(  \mathbb{Q}%
_{h}\right)  \right\Vert ^{2}. \label{2.33}%
\end{equation}
Then we calculate
\begin{equation}
\label{2.34}%
\begin{split}
&  \big\Vert s_{h}^{-1}S_{1h}u_{i};L^{2}\left(  \mathbb{Q}_{h}\right)
\big\Vert^{2}\\
&  \quad\le c\left(  h^{-2}\left\Vert u_{i}^{\bot};L^{2}\left(  \mathbb{Q}%
_{h}\right)  \right\Vert ^{2}+s_{h}\left(  y^{0}\right)  ^{-1}S_{1h}\left(
y^{0}\right)  \operatorname*{mes}\nolimits_{3}(\mathbb{Q}_{h})\left(
\left\vert \mathcal{U}_{i}^{0}\right\vert ^{2}+\left\vert \mathcal{U}%
_{6-i}^{0}\right\vert ^{2}+\left\vert \mathcal{U}_{0}^{0}\right\vert
^{2}\right)  \right) \\
&  \quad\le c\bigg(\left\Vert D\left(  \nabla\right)  u;L^{2}(\mathbb{Q}%
_{h})\right\Vert ^{2}+\left\Vert s_{h}^{-1}S_{1h}u_{i};L^{2}(\mathbb{Q}%
_{h}^{\prime})\right\Vert ^{2}+h^{2}\left\Vert s_{h}^{-1}S_{1h}\frac{\partial
u_{i}}{\partial z};L^{2}(\mathbb{Q}_{h}^{\prime})\right\Vert ^{2}\\
&  \qquad\qquad+h^{2}\left\Vert s_{h}^{-1}S_{1h}\frac{\partial u_{3}}{\partial
y_{i}};L^{2}(\mathbb{Q}_{h}^{\prime})\right\Vert ^{2}+\left\Vert \nabla_{y}
u^{\prime};L^{2}(\mathbb{Q}_{h}^{\prime})\right\Vert ^{2}\bigg)\\
&  \quad\le c\left(  \left\Vert D\left(  \nabla\right)  u;L^{2}\left(
\mathbb{Q}_{h}\right)  \right\Vert ^{2}+|||u;\mathbb{Q}_{h}^{\prime
}|||_{\bullet}^{2}\right)  .
\end{split}
\end{equation}
This calculation needs a detailed commentary. We here and further take into
account the following formulas for weights:
\begin{equation}
\label{2.35}%
\begin{split}
&  s_{h}\left(  y^{0}\right)  ^{-1}S_{hq}\left(  y^{0}\right)  \leq
\sup\limits_{x\in\mathbb{Q}_{h}}s_{h}\left(  y\right)  ^{-1}S_{hq}\left(
y\right)  \leq cs_{h}\left(  y^{0}\right)  ^{-1}S_{hq}\left(  y^{0}\right)
,\\
&  h^{2q}s_{h}\left(  y\right)  ^{-1}S_{hq}\left(  y\right)  \le c_{q},\qquad
q=0,1.
\end{split}
\end{equation}
The first inequality in (\ref{2.34}) was obtained by using decomposition
(\ref{2.25}) and a direct computation of the norm in $L^{2}(\mathbb{Q}_{h})$.
The factor $\operatorname*{mes}\nolimits_{3}\mathbb{Q}_{h}=h^{3}$ was
compensated due to the relations%
\[
\left\Vert u_{j};L^{2}(\mathbb{Q}_{h}^{\prime})\right\Vert ^{2}=h^{-3}%
\left\Vert \mathcal{U}_{j};L^{2}(\mathbb{Q}_{1}^{\prime})\right\Vert ^{2},
\qquad\left\Vert \frac{\partial u_{j}}{\partial x_{k}};L^{2}(\mathbb{Q}%
_{h}^{\prime})\right\Vert ^{2}=h^{-1}\left\Vert \frac{\partial\mathcal{U}_{j}%
}{\partial\xi_{k}};L^{2}(\mathbb{Q}_{1}^{\prime})\right\Vert ^{2}%
\]
but the last one still leaves the coefficient $h^{2}$ on the square of the
$L^{2}(\mathbb{Q}_{h}^{\prime})$-norm of a derivative. Then we applied
estimates (\ref{2.30})--(\ref{2.33}). Finally we recalled definition
(\ref{2.7}) of a weighted anisotropic norm.

Augmenting (\ref{2.35}) with the relation $S_{h1}\left(  y\right)  ^{-1}%
S_{h2}\left(  y\right)  \leq h^{-1}$ in $\mathbb{Q}_{h}$, we continue as
follows:
\[%
\begin{split}
&  h^{2}\left\Vert s_{h}^{-1}S_{2h}u_{3};L^{2}(\mathbb{Q}_{h})\right\Vert
^{2}\\
&  \qquad\leq c\left(  h^{-2}\left\Vert u_{3}^{\bot};L^{2}(\mathbb{Q}%
_{h})\right\Vert ^{2}+h^{2}s_{h}\left(  y^{0}\right)  ^{-1}S_{2h}\left(
y^{0}\right)  \operatorname*{mes}\nolimits_{3}\mathbb{Q}_{h}\left(  \left\vert
\mathcal{U}_{3}^{0}\right\vert ^{2}+\left\vert \mathcal{U}_{4}^{0}\right\vert
^{2}+\left\vert \mathcal{U}_{5}^{0}\right\vert ^{2}\right)  \right) \\
&  \qquad\leq c\bigg(\left\Vert D\left(  \nabla\right)  u^{\bot}%
;L^{2}(\mathbb{Q}_{h})\right\Vert ^{2}+h^{2}\left\Vert s_{h}^{-1}S_{2h}%
u_{3};L^{2}(\mathbb{Q}_{h}^{\prime})\right\Vert ^{2}+h^{2}\left\Vert
s_{h}^{-1}S_{1h}\nabla_{y}u_{3};L^{2}(\mathbb{Q}_{h}^{\prime})\right\Vert
^{2}\\
&  \qquad\qquad+h^{2}\left\Vert s_{h}^{-1}S_{1h}\partial_{z}u^{\prime}%
;L^{2}(\mathbb{Q}_{h}^{\prime})\right\Vert ^{2}\bigg)\\
&  \qquad\leq c\left(  \left\Vert D\left(  \nabla\right)  u;L^{2}%
(\mathbb{Q}_{h})\right\Vert ^{2}+|||u;\mathbb{Q}_{h}^{\prime}|||_{\bullet}%
^{2}\right)  ,
\end{split}
\]
where as usual $u^{\prime}=\left(  u_{1},u_{2}\right)  ^{\top}$. Since the
left $3\times3$-block of the rigid motion matrix (\ref{2.24}) is annulled by
differentiation, treating derivatives of $u_{j}$ becomes much simpler:
\begin{equation}%
\begin{split}
&  \left\Vert \frac{\partial u_{1}}{\partial y_{2}};L^{2}(\mathbb{Q}%
_{h})\right\Vert ^{2}+\left\Vert \frac{\partial u_{2}}{\partial y_{1}}%
;L^{2}(\mathbb{Q}_{h})\right\Vert ^{2}\leq c\left(  \left\Vert \nabla u^{\bot
};L^{2}(\mathbb{Q}_{h})\right\Vert ^{2}+h^{3}\left\vert \mathcal{U}_{6}%
^{0}\right\vert ^{2}\right) \\
&  \quad\leq c\bigg(\left\Vert D\left(  \nabla\right)  u;L^{2}(\mathbb{Q}%
_{h})\right\Vert ^{2}+\left\Vert \frac{\partial u_{1}}{\partial y_{2}}%
;L^{2}(\mathbb{Q}_{h}^{\prime})\right\Vert ^{2}+\left\Vert \frac{\partial
u_{2}}{\partial y_{1}};L^{2}(\mathbb{Q}_{h}^{\prime})\right\Vert ^{2}\bigg),\\
&  h^{2}\left\Vert s_{h}^{-1}S_{h1}\frac{\partial u^{\prime}}{\partial
z};L^{2}(\mathbb{Q}_{h})\right\Vert ^{2}+h^{2}\left\Vert s_{h}^{-1}%
S_{h1}\nabla_{y}u_{3};L^{2}(\mathbb{Q}_{h})\right\Vert ^{2}\\
&  \quad\leq c\left(  \left\Vert \nabla u^{\bot};L^{2}(\mathbb{Q}%
_{h})\right\Vert ^{2}+h^{2}s_{h}\left(  y^{0}\right)  ^{-2}S_{h1}\left(
y^{0}\right)  ^{2}h^{3}\left(  \left\vert \mathcal{U}_{4}^{0}\right\vert
^{2}+\left\vert \mathcal{U}_{5}^{0}\right\vert ^{2}\right)  \right) \\
&  \quad\leq c\bigg(\left\Vert D(\nabla)u;L^{2}(\mathbb{Q}_{h})\right\Vert
^{2}+h^{2}\Big\Vert s_{h}^{-1}S_{h1}\frac{\partial u^{\prime}}{\partial
z};L^{2}(\mathbb{Q}_{h}^{\prime})\Big\Vert^{2}+h^{2}\left\Vert s_{h}%
^{-1}S_{h1}\nabla_{y}u_{3};L^{2}(\mathbb{Q}_{h}^{\prime})\right\Vert
^{2}\bigg).
\end{split}
\label{2.36}%
\end{equation}
By the way, one may avoid to present (\ref{2.36}) down because the
$L^{2}\left(  \Omega_{h}\right)  $-norms of $\partial u_{1}/\partial y_{2}$
and $\partial u_{2}/\partial y_{1}$ had been estimated in (\ref{2.13}).
However, we observe that the derivatives $\partial u_{i}/\partial
y_{i}=\varepsilon_{ii}(u)$ and $\partial_{z}u_{3}=\varepsilon_{33}(u)$ figure
in the stress column (\ref{1.5})--(\ref{1.00}) and collect estimates obtained
above to arrive at the equality%
\[
|||u;\mathbb{Q}_{h}|||_{\bullet}^{2}\leq c\left(  \left\Vert D\left(
\nabla\right)  u;L^{2}(\mathbb{Q}_{h})\right\Vert ^{2}+|||u;\mathbb{Q}%
_{h}^{\prime}|||_{\bullet}^{2}\right)  .
\]
Summing these inequalities up over all cubes which are erected from cells of
the quadratic net of size $h$ in the plane $\mathbb{R}^{2}$ and have nonempty
intersection with the plate $\Omega_{h}$, yields the estimate%
\[
|||u;\Omega_{h}|||_{\bullet}^{2}\leq c\left(  \left\Vert D\left(
\nabla\right)  u;L^{2}\left(  \Omega_{h}\right)  \right\Vert ^{2}%
+|||u;\Omega_{h/2}|||_{\bullet}^{2}\right)  .
\]
We combine it with (\ref{2.000}) getting the result.
\end{proof}

\begin{remark}
As mentioned above, the optimality of distribution of the weights $h$ and
$s_{h}(y)$ in norms (\ref{2.2}) and (\ref{2.3}) is quite known. However,
relation (\ref{2.9}) shows that the norms $\left\Vert u_{j};L^{2}\left(
\mathbf{Q}_{hR}\right)  \right\Vert $ in the small cylinder $\mathbf{Q}%
_{hR}=\mathbb{B}_{hR}^{2}\times\left(  -h/2,h/2\right)  $ can be endowed with
the big factor $h^{-1}$ but weights in norm (\ref{2.7}) give the following
estimate only:%
\[
\left\vert \ln h\right\vert ^{-1}h^{-1}\left\Vert u_{j};L^{2}\left(
\mathbf{Q}_{hR}\right)  \right\Vert \leq c|||u;\mathbb{Q}_{h}|||_{\bullet}.
\]
In other words, it is worth to confirm impossibility of the change
$S_{hq}(y)\mapsto\left(  h^{2}+\left\vert y\right\vert ^{2}\right)  ^{-q/2}$
in (\ref{2.7}) with the simultaneous preservation of estimate (\ref{2.6}).

Let $\psi_{i}$ be smooth nontrivial functions such that $\psi_{i}\left(
t\right)  =0$ for $t\notin\left(  1/2,1\right)  $. We set $u_{i}\left(
x\right)  =\psi_{i}\left(  \left\vert \ln r\right\vert /\left\vert \ln
h\right\vert \right)  ,\ i=1,2,\ u_{3}\left(  x\right)  =0$ and obtain%
\[
\varepsilon_{il}\left(  u;x\right)  =\frac{1}{r}\frac{1}{\left\vert \ln
h\right\vert }\Psi_{il}\left(  \frac{\left\vert \ln r\right\vert }{\left\vert
\ln h\right\vert }\right)  ,\ \ i,l=1,2,\ \ \ \varepsilon_{13}\left(
u\right)  =\varepsilon_{23}\left(  u\right)  =\varepsilon_{33}\left(
u\right)  =0,
\]
where again $\Psi_{il}\left(  t\right)  =0$ for $t\notin(1/2,1)$, that is
$\Psi_{il}\left(  \left\vert \ln r\right\vert /\left\vert \ln h\right\vert
\right)  =0$ for\ $r\notin(h,\sqrt{h})$. We thus have%
\begin{equation}
\left\Vert D\left(  \nabla\right)  u;L^{2}\left(  \Omega_{h}\right)
\right\Vert ^{2}\leq\frac{c_{\Psi}h}{|\ln h|^{2}}\int_{h}^{\sqrt{h}}\frac
{rdr}{r^{2}}=\frac{c_{\Psi}h}{2|\ln h|}. \label{2.001}%
\end{equation}
At the same time,
\[%
\begin{split}
&  \left\Vert (h^{2}+|y|^{2})^{-1/2}u_{i};L^{2}\left(  \Omega_{h}\right)
\right\Vert ^{2}=2\pi h\int_{h}^{\sqrt{h}}\left\vert \psi_{i}\left(
\frac{|\ln r|}{|\ln h}\right)  \right\vert ^{2}\frac{rdr}{h^{2}+r^{2}}\\
&  \qquad\geq2\pi h|\ln h|\int_{1/2}^{1}\left\vert \psi_{i}\left(  \frac{|\ln
r|}{|\ln h|}\right)  \right\vert ^{2}d\frac{|\ln r|}{|\ln h|}=2\pi h|\ln
h|\left\Vert \psi_{i};L^{2}\left(  \tfrac{1}{2},1\right)  \right\Vert ^{2},\\
&  \left\Vert (h^{2}+\left\vert y\right\vert ^{2})^{-1/2}(1+|\ln
(h^{2}+\left\vert y\right\vert ^{2})|)^{-1}u_{i};L^{2}\left(  \Omega
_{h}\right)  \right\Vert ^{2}\\
&  \qquad=2\pi h\int_{h}^{\sqrt{h}}\left\vert \psi_{i}\left(  \frac{\left\vert
\ln r\right\vert }{\left\vert \ln h\right\vert }\right)  \right\vert
^{2}\left(  1+\left\vert \ln(h^{2}+\left\vert y\right\vert ^{2})\right\vert
\right)  ^{-2}\frac{rdr}{h^{2}+r^{2}}\leq\frac{c_{\Psi}h}{|\ln h|}.
\end{split}
\]
Glancing over these formulas convinces that logarithms cannot be eliminated in
the weighted norm (\ref{2.7}).
\end{remark}

\subsection{Traction-free edge of the plate with several support
areas\label{sect1.6}}

The approach applied above and described at length in the review paper
\cite{na398} helps to derive asymptotically exact weighted anisotropic
inequalities of Korn's type without requiring the lateral side $\upsilon_{h}$
of the plate to be clamped. Let us outline derivation of such inequalities.

To determine small clamped area on the lower base $\Sigma_{h}^{-}$, we fix
some points $y^{1},\dots,y^{J}$ inside $\omega$, $y^{j}\ne y^{k}$ for $j\ne
k$, and put
\begin{equation}
\Theta_{h}=\vartheta_{h}^{1}\cup\dots\cup\vartheta_{h}^{J},\ \ \ \vartheta
_{h}^{j}=\left\{  x:r_{j}:=\left\vert y-y^{j}\right\vert <Rh,\ z=-h/2\right\}
. \label{2.37}%
\end{equation}
Real supporting sets, of course, may be bigger, for instance $\theta_{h}%
^{j}\supset\vartheta_{h}^{j}$ like in (\ref{1.3}), but the Dirichlet condition%
\begin{equation}
u\left(  x\right)  =0,\ \ x\in\Theta_{h}, \label{2.38}%
\end{equation}
is sufficient for our purpose.

Since the resultant inequality is sensitive to the number of supporting sets
$\theta_{h}^{j}$, we focus on the case%
\begin{equation}
J\geq2. \label{2.99}%
\end{equation}
Notice that the Korn inequality remains the same for $J=2$ and the most
realistic case $J=3$. This and the peculiar case $J=1$ will be commented in
Remark \ref{remJ1}.

\begin{theorem}
\label{th2.3}If $J\geq2$ in (\ref{2.37}), then any displacement field $u\in
H_{0}^{1}\left(  \Omega_{h};\Theta_{h}\right)  ^{3}$ satisfies the weighted
anisotropic Korn inequality
\begin{equation}
|||u;\Omega_{h}|||_{\odot}\leq K_{\odot}\left(  \omega\right)  \left(
1+\left\vert \ln h\right\vert \right)  \left\Vert D\left(  \nabla\right)
u;L^{2}\left(  \Omega_{h}\right)  \right\Vert \label{2.55}%
\end{equation}
with a constant $K_{\odot}\left(  \omega\right)  $ independent of $h\in\left(
0,h_{0}\right]  $ and the weighted Sobolev norm%
\[%
\begin{split}
|||u;\Omega_{h}|||_{\odot}^{2}  &  =\int_{\Omega_{h}}\bigg[\sum_{i=1}%
^{2}\bigg(\left\vert \nabla_{y}u_{i}\right\vert ^{2}+h^{2}\mathbf{S}_{h1}%
^{2}\bigg(\left\vert \frac{\partial u_{i}}{\partial z}\right\vert
^{2}+\left\vert \frac{\partial u_{3}}{\partial y_{i}}\right\vert
^{2}\bigg)+\mathbf{S}_{h1}^{2}\left\vert u_{i}\right\vert ^{2}\bigg)\left\vert
\partial_{z}u_{3}\right\vert ^{2}\\
&  \qquad\qquad+h^{2}\mathbf{S}_{h2}^{2}\left\vert u_{3}\right\vert
^{2}\bigg]dx,
\end{split}
\]
where
\begin{equation}
\mathbf{S}_{hq}\left(  y\right)  =\max\left\{  S_{hq}\left(  y-y^{i}\right)
:j=1,\dots,J\right\}  \label{2.54}%
\end{equation}
and $S_{hq}$ are given in (\ref{2.5}).
\end{theorem}

\begin{proof}
We take a smooth displacement field $u$ satisfying (\ref{2.38}) and by the
same means as in Section \ref{sect2.2}, see, e.g., (\ref{2.9}), impose the
subsidiary conditions, cf. (\ref{2.0}),%
\begin{equation}
u(y,z)=0\text{ \ \ for \ }|y-y^{j}|<hR/2,\ \ \ j=1,\dots,J. \label{2.39}%
\end{equation}
First of all, we employ an elegant device from \cite{Shoikhet} and define the
vector function $\mathcal{U}$ with components%
\begin{equation}
\mathcal{U}_{i}\left(  y,\zeta\right)  =u_{i}\left(  y,h\zeta\right)
,\ \ i=1,2,\ \ \ \mathcal{U}_{3}\left(  y,\zeta\right)  =hu_{3}\left(
y,h\zeta\right)  \label{2.40}%
\end{equation}
in the vertically inflated plate $\Omega_{1}=\left\{  \left(  y,\zeta\right)
:y\in\omega,\ \left\vert \zeta\right\vert <1/2\right\}  $. The crucial
property of (\ref{2.40}) is expressed by the relation
\begin{equation}%
\begin{split}
&  \qquad\qquad\left\Vert D(\nabla)u;L^{2}\left(  \Omega_{h}\right)
\right\Vert ^{2}\\
&  =\int_{\Omega_{h}}\bigg[\sum_{i=1}^{2}\bigg(\left\vert \frac{\partial
u_{i}}{\partial y_{i}}\right\vert ^{2}+\frac{1}{2}\left\vert \frac{\partial
u_{i}}{\partial z}+\frac{\partial u_{3}}{\partial y_{i}}\right\vert
^{2}\bigg)+\frac{1}{2}\left\vert \frac{\partial u_{1}}{\partial y_{2}}%
+\frac{\partial u_{2}}{\partial y_{1}}\right\vert ^{2}+\left\vert
\frac{\partial u_{3}}{\partial z}\right\vert ^{2}\bigg]dydz\\
&  =h\int_{\Omega_{1}}\bigg[\sum_{i=1}^{2}\bigg(\left\vert \frac
{\partial\mathcal{U}_{i}}{\partial y_{i}}\right\vert ^{2}+\frac{1}{2}%
h^{-2}\left\vert \frac{\partial\mathcal{U}_{i}}{\partial\zeta}+\frac
{\partial\mathcal{U}_{3}}{\partial y_{i}}\right\vert ^{2}\bigg)+\frac{1}%
{2}\left\vert \frac{\partial\mathcal{U}_{1}}{\partial y_{2}}+\frac
{\partial\mathcal{U}_{2}}{\partial y_{1}}\right\vert ^{2}+h^{-4}\left\vert
\frac{\partial\mathcal{U}_{3}}{\partial z}\right\vert ^{2}\bigg]dyd\zeta\\
&  \geq h\int_{\Omega_{1}}\bigg[\sum_{i=1}^{2}\bigg(\left\vert \frac
{\partial\mathcal{U}_{i}}{\partial y_{i}}\right\vert ^{2}+\frac{1}%
{2}\left\vert \frac{\partial\mathcal{U}_{i}}{\partial\zeta}+\frac
{\partial\mathcal{U}_{3}}{\partial y_{i}}\right\vert ^{2}\bigg)+\frac{1}%
{2}\left\vert \frac{\partial\mathcal{U}_{1}}{\partial y_{2}}+\frac
{\partial\mathcal{U}_{2}}{\partial y_{1}}\right\vert ^{2}+\left\vert
\frac{\partial\mathcal{U}_{3}}{\partial z}\right\vert ^{2}\bigg]dyd\zeta\\
&  =h\left\Vert D\left(  \nabla_{y},\partial_{\zeta}\right)  \mathcal{U}%
;L^{2}\left(  \Omega_{1}\right)  \right\Vert ^{2}.
\end{split}
\label{2.41}%
\end{equation}
The term $\mathcal{U}^{\bot}$ in the decomposition, cf. (\ref{2.25}),%
\begin{equation}
\mathcal{U}\left(  y,\zeta\right)  =\mathcal{U}^{\bot}\left(  y,\zeta\right)
+d\left(  y,\zeta\right)  \mathcal{U}^{0} \label{2.42}%
\end{equation}
with the column%
\[
\mathcal{U}^{0}=\mathbf{d}\left(  \Omega_{1}\right)  ^{-1}\int_{\Omega_{1}%
}d\left(  y,\zeta\right)  ^{\top}\mathcal{U}\left(  y,\zeta\right)
dyd\zeta\in\mathbb{R}^{6}%
\]
meets the orthogonality conditions (\ref{2.28}) under the replacement
$\mathbb{Q}_{1}^{\prime}\mapsto\Omega_{1}$ and, therefore, using the Hardy and
Korn inequalities (\ref{2.16}) and (\ref{2.29}) yields
\begin{equation}%
\begin{split}
\left\Vert r_{j}^{-1}\left(  1+\left\vert \ln r_{j}\right\vert \right)
^{-1}\mathcal{U}^{\bot};L^{2}\left(  \Omega_{1}\right)  \right\Vert ^{2}  &
\leq c\left\Vert \mathcal{U}^{\bot};H^{1}\left(  \Omega_{1}\right)
\right\Vert ^{2}\\
&  \leq c\left\Vert D\left(  \nabla_{y},\partial_{\zeta}\right)
\mathcal{U}^{\bot};L^{2}\left(  \Omega_{1}\right)  \right\Vert ^{2}\\
&  =c\left\Vert D\left(  \nabla_{y},\partial_{\zeta}\right)  u_{j}%
;L^{2}\left(  \Omega_{1}\right)  \right\Vert ^{2}.
\end{split}
\label{2.43}%
\end{equation}
Now in view of (\ref{2.42}) and (\ref{2.39}), we write%
\begin{equation}
\mathbf{d}\big(\mathbf{Q}_{hR/2}^{j}\big)\mathcal{U}^{0}=-\int_{\mathbf{Q}%
_{hR/2}^{j}}d\left(  y,\zeta\right)  ^{\top}\mathcal{U}^{\bot}\left(
y,\zeta\right)  dyd\zeta=:\mathcal{F}^{j}\in\mathbb{R}^{6} \label{2.44}%
\end{equation}
where $\mathbf{Q}_{hR/2}^{j}=\left\{  \left(  y,\zeta\right)  :\left\vert
y-y^{j}\right\vert <hR/2,\ \left\vert \zeta\right\vert <1/2\right\}  $ is a
circular cylinder and, according to (\ref{2.43}), the right-hand side admits
the estimate%
\begin{equation}%
\begin{split}
\left\vert \mathcal{F}^{j}\right\vert  &  \leq c(\operatorname*{mes}%
\nolimits_{3}\mathbf{Q}_{hR/2}^{j})^{1/2}h\left(  1+\left\vert \ln
h\right\vert \right)  \left\Vert r_{j}^{-1}\left(  1+\left\vert \ln
r_{j}\right\vert \right)  ^{-1}U^{\bot};L^{2}\big(\mathbf{Q}_{hR/2}%
^{j}\big)\right\Vert ^{2}\\
&  \leq ch^{2}\left(  1+|\ln h|\right)  \left\Vert D\left(  \nabla
_{y},\partial_{z}\right)  U;L^{2}\left(  \Omega_{1}\right)  \right\Vert ^{2}%
\end{split}
\label{2.45}%
\end{equation}
Following \cite{na212}, see also \cite[\S 3.4]{na398}, we sum up equations
(\ref{2.44}), $j=1,\dots,J$, and obtain the linear algebraic system%
\begin{equation}
\mathcal{M}(h)U^{0}=\mathcal{F}:=\mathcal{F}^{1}+\dots+\mathcal{F}^{J}%
\in\mathbb{R}^{6} \label{2.46}%
\end{equation}
with the $6\times6$-matrix%
\begin{equation}
\mathcal{M}\left(  h\right)  =\mathbf{d}\big(\mathbf{Q}_{hR/2}^{1}%
\big)+\dots+\mathbf{d}\big(\mathbf{Q}_{hR/2}^{J}\big). \label{2.47}%
\end{equation}

Each of summands in (\ref{2.47}) is a Gram matrix, symmetric and positive
definite. Moreover, by virtue of (\ref{2.24}) and (\ref{2.27}), a simple
calculation shows that%
\begin{equation}
\mathbf{d}\big(\mathbf{Q}_{hR/2}^{j}\big)=\frac{\pi}{4}h^{2}R^{2}\left(
d\left(  y^{j},0\right)  ^{\top}d\left(  y^{j},0\right)  +\frac{1}%
{12}\mathsf{t}^{\top}\mathsf{t}+O\left(  h\right)  \right)  \label{2.48}%
\end{equation}
where $O\left(  h\right)  $ stands for a $6\times6$-matrix of size $6\times6$
\ with the natural norm of order $h$ and%
\begin{equation}
\mathsf{t}=\left(
\begin{array}
[c]{cccccc}%
0 & 0 & 0 & 1 & 0 & 0\\
0 & 0 & 0 & 0 & 1 & 0\\
0 & 0 & 0 & 0 & 0 & 0
\end{array}
\right)  ,\ \ \ \dfrac{1}{12}=\int_{-1/2}^{1/2}\zeta^{2}d\zeta,\ \ \ \frac
{\pi}{4}h^{2}R^{2}=\operatorname*{mes}\nolimits_{3}\mathbf{Q}_{hR/2}^{j}.
\label{2.488}%
\end{equation}
We observe that under requirement (\ref{2.99}) the inverse of matrix
(\ref{2.47}) satisfies the estimate%
\begin{equation}
\left\Vert \mathcal{M}\left(  h\right)  ^{-1};\mathbb{R}^{6\times6}\right\Vert
\leq ch^{-2}. \label{2.49}%
\end{equation}

In fact, the $6\times6$-matrix $M^{j}=d\left(  y^{j},0\right)  ^{\top}d\left(
y^{j},0\right)  +\tfrac{1}{12}\mathsf{t}^{\top}\mathsf{t}$ on the right of
(\ref{2.48}), is symmetric but only positive. However, in view of
(\ref{2.47}), estimate (\ref{2.49}) is a direct consequence of the following
fact:%
\begin{equation}
b\in\mathbb{R}^{6}\text{ \ \ and \ \ }b^{\top}M^{j}b=0,\ j=1,\dots
,J\ \ \Longrightarrow\ \ b=0. \label{2.50}%
\end{equation}
Moreover, the premise in (\ref{2.50}) is equivalent to
\begin{equation}
\mathsf{t}b=0\text{ \ and \ }d\left(  y^{j},0\right)  b=0,\ \ \ j=1,\dots,J.
\label{2.51}%
\end{equation}
We put the coordinate origin $y=0$ at the point $y^{1}$ and direct the $y_{1}%
$-axis through $y^{2}$ so that $y_{1}^{2}>0$ and $y_{2}^{2}=0$. Owing to
(\ref{2.488}) and (\ref{2.24}) the equalities $\mathsf{t}b=0$ and $d\left(
y^{1},0\right)  b=0$ in (\ref{2.51}) assure that $b_{4}=b_{5}=0$ and
$b_{1}=b_{2}=b_{3}=0$. Furthermore,%
\[
d\left(  y^{2},0\right)  -d\left(  y^{1},0\right)  =\left(
\begin{array}
[c]{cccccc}%
0 & 0 & 0 & 0 & 0 & 0\\
0 & 0 & 0 & 0 & 0 & -y_{1}^{2}\\
0 & 0 & 0 & 0 & y_{2}^{2} & 0
\end{array}
\right)
\]
and, therefore, $b_{6}=0$ due to (\ref{2.51}) with $j=1$ and $j=2$. So
(\ref{2.49}) is proved.

\bigskip

From (\ref{2.46}) together with (\ref{2.49}) and (\ref{2.45}) we derive that
\begin{equation}
\left\vert \mathcal{U}^{0}\right\vert \leq c\left(  1+\left\vert \ln
h\right\vert \right)  \left\Vert D\left(  \nabla_{y},\partial_{z}\right)
\mathcal{U};L^{2}\left(  \Omega_{1}\right)  \right\Vert . \label{2.52}%
\end{equation}
Thus, representation (\ref{2.42}), estimates (\ref{2.43}), (\ref{2.52}) and
relations (\ref{2.41}), (\ref{2.40}) between $\mathcal{U}$\ and $u$ give the
following Korn inequality%
\begin{equation}
|||u;\Omega_{h}|||^{2}\leq c\left(  1+\left\vert \ln h\right\vert \right)
^{2}\left\Vert D\left(  \nabla\right)  u;L^{2}\left(  \Omega_{h}\right)
\right\Vert ^{2}, \label{2.53}%
\end{equation}
with the only exception: after returning to $x$ and $u$ the right-hand side of
(\ref{2.41}) gains the integral $h^{4}\int\left\vert \partial_{z}%
u_{3}\right\vert ^{2}dx$ instead of $\int\left\vert \partial_{z}%
u_{3}\right\vert ^{2}dx$ as in (\ref{2.2}). This gap is filled easily because,
according to (\ref{1.5}), (\ref{1.7}), the expression $\left\Vert D\left(
\nabla\right)  u;L^{2}\left(  \Omega_{h}\right)  \right\Vert ^{2}$ on the
right of (\ref{2.53}) includes $\int\left\vert \varepsilon_{33}\left(
u\right)  \right\vert ^{2}dx=\int\left\vert \partial_{z}u_{3}\right\vert
^{2}dx$.

A repetition of calculations in Section \ref{sect2.2} (word by word but with
$s_{h}=1$, the basic relations (\ref{2.19}), (\ref{2.35}) and (\ref{2.39})
being preserved) allows us to introduce the weights $\mathbf{S}_{hq}$ given in
(\ref{2.54}) into the norm. Note that the real reason to put $s_{h}=1$ is the
absence of the Dirichlet condition (\ref{1.14}) on the lateral side
$\upsilon_{h}$ and the consequent impossibility to apply the Hardy
inequalities (\ref{2.15}), (\ref{2.22}).

A completion argument again completes the proof.
\end{proof}

\begin{remark}
The Dirichlet condition (\ref{2.38}) at small support zones, see (\ref{2.37})
and (\ref{2.99}), cannot maintain the Korn inequality without the factor
$1+\left\vert \ln h\right\vert $ as in (\ref{2.55}). To corroborate this
statement, we use the test function $u_{i}\left(  x\right)  =%
{\textstyle\prod\nolimits_{j=1}^{J}}
\chi\left(  \left\vert \ln\left(  r_{j}/R\right)  \right\vert /\left\vert \ln
h\right\vert \right)  $ where $\chi$ is taken from (\ref{2.11}) so that
$u_{i}\left(  x\right)  =1$ if all $r_{j}=\left\vert y-y_{j}\right\vert
>\sqrt{h}R$, and $u_{i}\left(  x\right)  =0$ if some $r_{j}<hR$. Clearly, for
$u=e_{i}u_{i}$, we obtain%
\[
\left\Vert u;L^{2}\left(  \Omega_{h}\right)  \right\Vert ^{2}=h^{1/2}%
(\left\vert \omega\right\vert ^{1/2}+O\left(  h\right)  ),\ \ \ C\geq
\left\Vert \mathbf{S}_{h1}u^{\prime};L^{2}\left(  \Omega_{h}\right)
\right\Vert ^{2}\geq c>0,
\]
and, similarly to (\ref{2.001}),%
\[
\left\Vert D\left(  \nabla\right)  u;L^{2}\left(  \Omega_{h}\right)
\right\Vert ^{2}\leq\frac{c_{\chi}}{R^{2}\left\vert \ln h\right\vert ^{2}}%
\int_{h}^{\sqrt{h}}\frac{rdr}{r^{2}}\leq\frac{c}{\left\vert \ln h\right\vert
}.
\]
The desired inference follows.
\end{remark}

\begin{remark}
\label{remJ1}If $J=1$ and $y^{1}=0$, the matrix $\mathcal{M}\left(  h\right)
=\mathbf{d}(\mathbf{Q}_{hR/2}^{1})$ in (\ref{2.47}) becomes%
\[
\pi^{2}h^{2}R^{2}\operatorname*{diag}\left\{  1,1,1,\tfrac{1}{4}\left(
\tfrac{1}{4}+h^{2}R^{2}\right)  ,\tfrac{1}{4}\left(  \tfrac{1}{4}+h^{2}%
R^{2}\right)  ,\tfrac{1}{2}h^{2}R^{2}\right\}
\]
and gets the right-hand bottom entry $O\left(  h^{4}\right)  $. In this way
(\ref{2.49}) loses validity and estimate (\ref{2.52}) alters crucially for
$U_{6}^{0}$. Thus, the resultant inequality requires a serious modification,
too (cf. \cite{na331} and \cite[\S 3.4 and \S 5.2]{na398}). We mention the
displacement field $u\left(  x\right)  =\left(  1-\chi\left(  \tfrac{r}%
{2hR}\right)  \right)  \left(  -y_{2},y_{1},0\right)  ^{\top}$ which satisfies
the relations%
\[
h^{-1/2}\left\Vert u;L^{2}\left(  \Omega_{h}\right)  \right\Vert \geq
c>0,\ \ \ h^{-3/2}\left\Vert D\left(  \nabla\right)  u;L^{2}\left(  \Omega
_{h}\right)  \right\Vert \leq C
\]
and indicates a power-law growth of the Korn constant as $h\rightarrow+0$.
\end{remark}

\section{The convergence theorem and the three-dimensional boundary
layer\label{sect3}}

\subsection{The Kirchhoff model with the Sobolev point
condition\label{sect3.1}}

In \cite{ButCaNa2} we will present a detailed procedure of dimension reduction
with turns the elasticity problem (\ref{1.11})--(\ref{1.14}) into the
two-dimensional Kirchhoff model of an anisotropic plate (\ref{1.1}) clamped
over the lateral side $\upsilon_{h}$ and the small area $\theta_{h}$, cf.
(\ref{1.13}) and (\ref{1.14}). We however start with a classical and simple
result on convergence which can be achieved by any of known methods, cf.
monographs \cite{Ciarlet1, Ciarlet2, LeDret, Nabook, SaPa} and other literature.

We assume the following representation\footnote{In \cite{ButCaNa2} this
assumption will be weakned quite much.} for the right-hand side in
(\ref{1.11}) :%
\begin{equation}
f_{i}\left(  h,y,z\right)  =h^{-1/2}g_{i}(y),\ i=1,2,\ \ \ \ f_{3}\left(
h,y,z\right)  =h^{1/2}g_{3}(y) \label{A4}%
\end{equation}
with $g=(g_{1},g_{2},g_{3})\in L^{2}\left(  \omega\right)  ^{3}.$

If the rigidity matrix $A$ has the form (\ref{1.10}) in the isotropic Hooke's
law (\ref{1.9}), the average longitudinal displacements $w^{\prime}=\left(
w_{1},w_{2}\right)  ^{\top}$ is solution of the two-dimensional elasticity
system%
\begin{equation}
-\mu\Delta_{y}w^{\prime}(y)-(\lambda^{\prime}+\mu)\nabla_{y}\nabla_{y}^{\top
}w^{\prime}(y)=g^{\prime}(y),\qquad y\in\omega, \label{A1}%
\end{equation}
with $g^{\prime}=\left(  g_{1},g_{2}\right)  ^{\top}$ and the deflexion
$w_{3}$ is solution of the bi-harmonic equation%
\begin{equation}
\frac{\mu}{3}\frac{\lambda+\mu}{\lambda+2\mu}\Delta_{y}^{2}w_{3}\left(
y\right)  =g_{3}\left(  y\right)  ,\qquad y\in\omega. \label{A2}%
\end{equation}
Here, $\nabla_{y}=(\partial/\partial y_{1},\partial/\partial y_{2})^{\top}$,
$\Delta_{y}=\nabla_{y}^{\top}\nabla_{y}$ is the Laplace operator in the
$y$-variables, and the coefficient%
\begin{equation}
\lambda^{\prime}=\frac{2\lambda\mu}{\lambda+2\mu} \label{A3}%
\end{equation}
is computed through the Lam\'{e} constants $\lambda\geq0$ and $\mu\geq0$.

The condition (\ref{1.13}) on the lateral side $\upsilon_{h}$, see
(\ref{1.2}), requires for the Dirichlet condition on the contour
$\partial\omega$%
\begin{equation}
w_{i}(y)=0,\ i=1,2,\qquad w_{3}(y)=0,\qquad\partial_{n}w_{3}(y)=0,\qquad
y\in\partial\omega, \label{A5}%
\end{equation}
where $\partial_{n}=n^{\top}\nabla_{y}$ and $n=(n_{1},n_{2})^{\top}$ is the
unit vector of the outward normal at $\partial\omega$. Finally, the support
area $\theta_{h}$, see (\ref{1.3}) and (\ref{1.13}), is reflected by the
Sobolev point condition%
\begin{equation}
w_{3}\left(  \mathcal{O}\right)  =0. \label{A6}%
\end{equation}
It is well known, see, e.g., \cite{BuNa}, that for any $g\in L^{2}\left(
\omega\right)  ^{3}$ the problem (\ref{A1}), (\ref{A2}), (\ref{A5}),
(\ref{A6}) has a unique generalized solution $w\in H_{0}^{1}\left(
\omega\right)  ^{2}\times H_{0}^{2}\left(  \omega\right)  $ while $w_{1}%
,w_{2}\in H^{2}\left(  \omega\right)  $ and $w_{3}\in H_{loc}^{4}\left(
\overline{\omega}\setminus\mathcal{O}\right)  $ but in general $w_{3}\notin
H^{3}\left(  \omega\right)  $. The next convergence theorem holds and it can
be proved by an impalpable modification of the standard approach for a plate
without the small support area $\theta_{h}$.

\begin{theorem}
\label{th3.3}The rescaled displacements $h^{3/2}u_{3}(h,y,h\zeta)$ and
$h^{1/2}u_{i}(h,y,h\zeta)$, $i=1,2$, in the three-dimensional problem
(\ref{1.11})--(\ref{1.14}) with the right-hand side (\ref{A4}) converge in
$L^{2}(\omega\times(-1/2,1/2))$ as $h\to+0$ to the functions $w_{3}(y)$ and
$w_{i}(y)-\zeta\dfrac{\partial w_{3}}{\partial y_{i}}(y)$, $i=1,2$,
respectively, where $\zeta=h^{-1}z$ is the stretched coordinate and $w=\left(
w_{1},w_{2},w_{3}\right)  ^{\top}\in H_{0}^{1}\left(  \omega\right)
^{2}\times H_{0}^{2}\left(  \omega\right)  $ is a solution of the
two-dimensional Dirichlet-Sobolev problem (\ref{A1}), (\ref{A2}), (\ref{A5}),
(\ref{A6}).
\end{theorem}

The main unfavorable result of \cite{ButCaNa2} reads: the convergence rate in
Theorem \ref{th3.3} is unacceptably low, namely $O(\left\vert \ln h\right\vert
^{-1/2})$ and thereafter we prepare for an elaboration of the asymptotic
structures of elastic fields in $\Omega_{h}$ near $\theta_{h}$.

\subsection{Sketch of asymptotic expansions in a thin plate.\label{sect3.2}}

Theorem \ref{th3.3} remains valid for an anisotropic plate, i.e., with
arbitrary rigidity matrix $A$ in the Hooke's law, where the differential
equations (\ref{A1}) and (\ref{A2}) are replaced with
\begin{align}
\mathcal{L}^{\prime}\left(  \nabla_{y}\right)  w^{\prime}(y)  &  =g^{\prime
}(y),\ \ \ y\in\omega,\label{A7}\\
\mathcal{L}_{3}\left(  \nabla_{y}\right)  w_{3}(y)  &  =g_{3}(y),\ \ \ y\in
\omega, \label{A8}%
\end{align}
where the $2\times2$-matrix $\mathcal{L}^{\prime}$ of second-order operators
and the scalar fourth-order operator $\mathcal{L}_{3}$ are given by
\begin{align}
\mathcal{L}^{\prime}\left(  \nabla_{y}\right)   &  =\mathcal{D}^{\prime
}\left(  -\nabla_{y}\right)  ^{\top}\mathcal{A}^{0}\mathcal{D}^{\prime}\left(
\nabla_{y}\right)  ,\label{A9}\\
\mathcal{L}_{3}\left(  \nabla_{y}\right)   &  =\frac{1}{6}\mathcal{D}%
_{3}\left(  \nabla_{y}\right)  ^{\top}\mathcal{A}^{0}\mathcal{D}_{3}\left(
\nabla_{y}\right)  . \label{A99}%
\end{align}
Here, the symmetric and positive definite $3\times3$-matrix $A^{0}$ is
computed as follows:%
\[
A^{0}=A_{(yy)}-A_{(yz)}A_{(zz)}^{-1}A_{(zy)},\ \ \ \ A=\left(
\begin{array}
[c]{cc}%
A_{(yy)} & A_{(yz)}\\
A_{(zy)} & A_{(zz)}%
\end{array}
\right)
\]
and%
\begin{equation}
\mathcal{D}^{\prime}\left(  \nabla_{y}\right)  =\left(
\begin{array}
[c]{ccc}%
\partial_{1} & 0 & 2^{-1/2}\partial_{2}\\
0 & \partial_{2} & 2^{-1/2}\partial_{1}%
\end{array}
\right)  ^{\top},\ \ \ \mathcal{D}_{3}\left(  \nabla_{y}\right)  =\left(
2^{-1/2}\partial_{1}^{2},2^{-1/2}\partial_{2}^{2},\partial_{1}\partial
_{2}\right)  ^{\top}. \label{A11}%
\end{equation}
Calculations of asymptotic expansions and the differential operators
(\ref{A9}), which are entirely adapted to the Mandel-Voigt notation, will be
presented in \cite{ButCaNa2} but also can be derived by any asymptotic
procedure for an asymptotic analysis of thin plates, for example, \cite{na273}
and \cite[\S 4]{Nabook}. In Section 3 we will apply the ordinary asymptotic
expansion of the solution to the problem (\ref{1.11})--(\ref{1.14}) with the
right-hand side (\ref{A4})%
\begin{equation}
u(h,x)\sim h^{-3/2}\sum_{p=0}^{3}h^{p}W^{p}\left(  \zeta,\nabla_{y}\right)
w(y)+\dots\label{A12}%
\end{equation}
written in an unusual form, i.e. with the help of the following $3\times
3$-matrix differential operators%
\begin{equation}%
\begin{split}
&  W^{0}\left(  \zeta,\nabla_{y}\right)  =\left(
\begin{array}
[c]{ccc}%
0 & 0 & 0\\
0 & 0 & 0\\
0 & 0 & 1
\end{array}
\right)  ,\qquad W^{1}\left(  \zeta,\nabla_{y}\right)  =\left(
\begin{array}
[c]{ccc}%
1 & 0 & -\zeta\partial_{1}\\
0 & 1 & -\zeta\partial_{2}\\
0 & 0 & 0
\end{array}
\right)  ,\\
&  W^{2}\left(  \zeta,\nabla_{y}\right)  =\mathbb{J}^{-1}A_{(zz)}^{-1}%
A_{(zy)}\left(  -\zeta\mathbb{I}_{3},2^{1/2}\left(  \frac{\zeta^{2}}{2}%
-\frac{1}{24}\right)  \mathbb{I}_{3}\right)  \mathcal{D}\left(  \nabla
_{y}\right)
\end{split}
\label{3.53}%
\end{equation}
where $\mathbb{I}_{3}=\mathrm{diag}\{1,1,1\}$ and $\mathbb{J}=\mathrm{diag}%
\{2^{-1/2},2^{-1/2},1\}$ are diagonal matrices and $\mathcal{D}\left(
\nabla_{y}\right)  $ is the $6\times3$-matrix composed of the blocks
(\ref{A11})%
\[
\mathcal{D}\left(  \nabla_{y}\right)  =\left(
\begin{array}
[c]{cc}%
\mathcal{D}^{\prime}\left(  \nabla_{y}\right)  & \mathbb{O}_{3\times1}\\
\mathbb{O}_{3\times2} & \mathcal{D}_{3}\left(  \nabla_{y}\right)
\end{array}
\right)
\]
and the null matrices $\mathbb{O}_{p\times q}$ of size $p\times q$. Coherently
with (\ref{A13}), the matrix differential operator $\mathcal{L}(\nabla_{y})$
of the affiliated system (\ref{A7}), (\ref{A8}) involves the following
block-diagonal matrix $\mathcal{A}$ of size $6\times6$,%
\begin{equation}
\mathcal{A}=\mathrm{diag}\{A^{0},\frac{1}{6}A^{0}\},\ \ \ \mathcal{L}\left(
\nabla\right)  =\mathrm{diag}\{\mathcal{L}^{\prime}\left(  \nabla_{y}\right)
,\mathcal{L}_{3}\left(  \nabla_{y}\right)  \}. \label{A13}%
\end{equation}

Let us hint at the choice of operators (\ref{3.53}) in (\ref{A12}) even if a
detailed description of the dimension reduction procedure will be given in
\cite{ButCaNa2}. The differential operators $L\left(  \nabla\right)  $ and
$N^{\pm}\left(  \nabla\right)  $ on the left in (\ref{1.11}) and (\ref{1.12})
admit the decompositions%
\begin{equation}%
\begin{split}
L(\nabla)  &  =h^{-2}L^{0}\left(  \partial_{\zeta}\right)  +h^{-1}L^{1}\left(
\nabla_{y},\partial_{\zeta}\right)  +h^{0}L^{2}\left(  \nabla_{y}\right)  ,\\
N^{\pm}(\nabla)  &  =h^{-1}N^{0\pm}\left(  \partial_{\zeta}\right)
+h^{0}N^{1\pm}\left(  \nabla_{y}\right)  ,
\end{split}
\label{3.4}%
\end{equation}
where $\zeta=h^{-1}z\in(-1/2,1/2)$ is the stretched coordinate, $\partial
_{\zeta}=\partial/\partial\zeta$, and%
\begin{equation}%
\begin{split}
L^{0}\left(  \partial_{\zeta}\right)   &  =D(0,0,-\partial_{\zeta})^{\top
}AD(0,0,\partial_{\zeta}),\qquad L^{2}\left(  \nabla_{y}\right)
=D(-\nabla_{y},0)^{\top}AD(\nabla_{y},0),\\
L^{1}\left(  \nabla_{y},\partial_{\zeta}\right)   &  =D(0,0,-\partial_{\zeta
})^{\top}AD(\nabla_{y},0)+D(-\nabla_{y},0)^{\top}AD(0,0,\partial_{\zeta}),\\
N^{0\pm}\left(  \partial_{\zeta}\right)   &  =D(\pm e_{3})^{\top
}AD(0,0,\partial_{z}),\qquad N^{1\pm}\left(  \nabla_{y}\right)  =D(\pm
e_{3})^{\top}AD(\nabla_{y},0).
\end{split}
\label{3.5}%
\end{equation}
Then we have%
\begin{equation}%
\begin{split}
&  L\sum_{p=0}^{3}h^{p}W^{p}=h^{-2}L^{0}W^{0}+h^{-1}\left(  L^{0}W^{1}%
+L^{1}W^{0}\right) \\
&  \qquad\qquad\qquad+h^{0}\left(  L^{0}W^{2}+L^{1}W^{1}+L^{2}W^{0}\right)
+h^{1}\left(  L^{0}W^{3}+L^{1}W^{2}+L^{2}W^{1}\right) \\
&  \qquad\qquad\qquad+h^{2}\left(  L^{1}W^{3}+L^{2}W^{2}\right)  +h^{3}%
L^{2}W^{3}=:\sum_{q=0}^{5}h^{q-2}F^{q},\\
&  N^{\pm}\sum_{p=0}^{3}h^{p}W^{p}=h^{-1}N^{0\pm}W^{0}+h^{0}\left(  N^{0\pm
}W^{1}+N^{1\pm}W^{0}\right) \\
&  \qquad\qquad\qquad+h^{1}\left(  N^{0\pm}W^{2}+N^{1\pm}W^{1}\right)
+h^{2}\left(  N^{0\pm}W^{3}+N^{1\pm}W^{2}\right) \\
&  \qquad\qquad\qquad+h^{3}+N^{1\pm}W^{3}=:\sum_{q=0}^{4}h^{q-1}G^{q\pm}.
\end{split}
\label{A14}%
\end{equation}
The operators (\ref{3.53}) are selected such that%
\begin{equation}
F^{q}=0,\ G^{q\pm}=0\text{ with }q=0,1,2. \label{A15}%
\end{equation}
It is not possible to annul both $F^{3}$ and $G^{3\pm}$ but $W^{3}\left(
\zeta,\nabla_{y}\right)  $ is fixed such that%
\begin{equation}
\left(  F_{1}^{3},F_{2}^{3}\right)  ^{\top}=\mathcal{L}^{\prime}\left(
\nabla_{y}\right)  \text{ \ \ and \ \ }F_{3}^{3}=0,\ G^{3\pm}=0. \label{A16}%
\end{equation}
Moreover,%
\begin{equation}
\int_{-1/2}^{1/2}F_{3}^{4}d\zeta+G_{3}^{4\pm}+G_{3}^{4-}=\mathcal{L}%
_{3}\left(  \nabla_{y}\right)  . \label{A17}%
\end{equation}
The differential operators $W^{p}\left(  \zeta,\nabla_{y}\right)  $ in
(\ref{A12}), (\ref{3.53}) as well as (\ref{A9}) in (\ref{A16}), (\ref{A17})
are defined uniquely through the matrices $D(\nabla)$ and $A$ in (\ref{1.7})
and (\ref{1.9}). For general elliptic problem, an algebraic procedure to
construct asymptotics type (\ref{A12}) and the resultant operator
$\mathcal{L}(\nabla_{y})$ in (\ref{A13}) is developed in \cite{na192}, see
also \cite[\S 2 Ch.5]{Nabook} and the review \cite{na262}. This procedure
serves for constructing asymptotic expansions in thin domains and in unbounded
domains with cylindrical and layer-shaped outlets to infinity, cf.
\cite[\S 2]{ButCaNa2} and Section \ref{sect3.4} below.

Finally it should be mentioned that explicit formulas for $W^{3}\left(
\zeta,\nabla_{y}\right)  w(y)$ and higher order terms in (\ref{A12}) are
available but are of no further use.

\subsection{Boundary layer effects\label{sect3.3}}

Although Theorem \ref{th3.3} singles out the solution $w$ of the limit problem
(\ref{A7}), (\ref{A8}), (\ref{A5}), (\ref{A6}) in $\omega$ as a limit of the
solution $u$ of the original problem in the thin plate $\Omega_{h}$, the
Dirichlet conditions (\ref{1.13}) at the small support zone $\theta_{h}$
maintain only the point condition (\ref{A6}) for the deflection $w_{3}$ and
are not reflected in the problem for the longitudinal displacements
$w^{\prime}$. At the same time, the components $w_{1}$ and $w_{2}$ of
$w^{\prime}$ leave small discrepancies in conditions (\ref{1.13}) unless
accidentally. As usual, a variation of boundary conditions on a set with a
small diameter brings boundary layer effects, see \cite[Ch.5]{MaNaPl} for
general apprehension. However, due to the assumed comparability of the plate
thickness $h$ and diameters of $\omega_{h}^{j}$ the boundary layer in problem
(\ref{1.11})--(\ref{1.14}) exhibits a very specific and intricate structure.
To describe it, we introduce the stretched coordinates (\ref{3.35}) and
observe that changing $x\mapsto\xi$ and putting $h=0$ transform the domains
$\Omega_{h}$ and $\theta_{h}$ into the layer (\ref{3.36}) between the planes
$\Sigma_{\pm}=\mathbb{R}^{2}\times\left\{  \pm1/2\right\}  $ and the set
$\theta\subset\Sigma_{-}$ of unit size respectively. We also denote
$\Sigma_{\bullet}=\Sigma_{-}\setminus\overline{\theta}$.

Let us recall the following weighted anisotropic Korn inequality which is
proved in \cite{na155}, see also \cite{na243} and \cite{Nabook}, and looks
quite similar to (\ref{2.55}) and (\ref{2.6}).

\begin{lemma}
\label{lem3.4}For any smooth and compactly supported vector function $v$
satisfying the Dirichlet condition on $\theta$ the inequality%
\begin{equation}
\left\Vert u;V_{0}^{1}\left(  \Lambda;\theta\right)  \right\Vert \leq
c\left\Vert D\left(  \nabla_{\xi}\right)  v;L^{2}\left(  \Lambda\right)
\right\Vert \label{3.K}%
\end{equation}
is valid, where
\begin{align}
\left\Vert u;V_{0}^{1}(\Lambda;\theta)\right\Vert ^{2}  &  =\int_{\Lambda
}\Bigg[\sum_{i=1}^{2}\left(  \left\vert \nabla_{\xi}v_{i}\right\vert
^{2}+S_{1}^{2}\left(  \left\vert \frac{\partial v_{i}}{\partial\zeta
}\right\vert ^{2}+\left\vert \frac{\partial v_{3}}{\partial\eta_{i}
}\right\vert ^{2}+\left\vert v_{i}\right\vert ^{2}\right)  \right)
\label{3.41}\\
&  \ \ \ \ \ \ \ \ \ \ \ \ \ \ \ \ \ \ +\left\vert \partial_{\zeta}%
v_{3}\right\vert ^{2}+S_{2}^{2}\left\vert v_{3}\right\vert ^{2}\Bigg]\,d\xi
,\nonumber
\end{align}
the weighted space $V_{0}^{1}\left(  \Lambda;\theta\right)  $ is a completion
of $C_{c}^{\infty}\left(  \overline{\Lambda}\setminus\overline{\theta}\right)
^{3}$ with respect to the norm (\ref{3.41}) and
\[
S_{k}(\eta)=\left(  1+\rho^{2}\right)  ^{-k/2}\left(  1+\ln\left(  1+\rho
^{2}\right)  \right)  ^{-1},\qquad\rho=|\eta|.
\]

\end{lemma}

Our aim is to investigate the elasticity problem%
\begin{gather}
-D\left(  -\nabla_{\xi}\right)  ^{\top}AD\left(  \nabla_{\xi}\right)  v\left(
\xi\right)  =F\left(  \xi\right)  ,\qquad\xi\in\Lambda,\label{3.38}\\%
\begin{cases}
D\left(  e_{3}\right)  ^{\top}AD\left(  \nabla_{\xi}\right)  v\left(
\xi\right)  =G^{+}(\xi), & \quad\xi\in\Sigma^{+},\\
D\left(  -e_{3}\right)  ^{\top}AD\left(  \nabla_{\xi}\right)  v(\xi
)=G^{\bullet}(\xi), & \quad\xi\in\Sigma^{\bullet},
\end{cases}
\label{3.39}\\
v\left(  \xi\right)  =0,\qquad\xi\in\theta, \label{3.40}%
\end{gather}
in layer (\ref{3.36}) clamped along the area $\theta$; we assume that%
\begin{equation}%
\begin{split}
&  S_{1}^{-1}F_{i}\in L^{2}\left(  \Lambda\right)  ,\qquad S_{1}^{-1}%
G_{i}^{\pm}\in L^{2}\left(  \Sigma^{\pm}\right)  ,\qquad i=1,2,\\
&  F_{3}(\eta,\zeta)=F_{3}^{0}(\eta,\zeta)+F_{3}^{1}(\eta),\qquad\int
_{-1/2}^{1/2}F_{3}^{0}(\eta,\zeta)d\zeta+\sum_{\pm}G_{3}^{\pm}(\eta)=0,\\
&  S_{2}^{-1}F_{3}^{0}\in L^{2}\left(  \mathbb{R}^{2}\right)  ,\qquad
F_{3}^{0}\in L^{2}(\Lambda),\qquad G_{3}^{\pm}\in L^{2}\left(  \Sigma^{\pm
}\right)  ,
\end{split}
\label{3.FG}%
\end{equation}
where $G^{-}$ is an extension of $G^{\bullet}$ over $\theta$ (observe that
$u=0$ on $\theta)$.

By Lemma \ref{lem3.4} we can give a variational formulation to problem
(\ref{3.38})-(\ref{3.40}).

Multiplying (\ref{3.38}) scalarly with a test function $u\in C_{c}^{\infty
}\left(  \overline{\Lambda}\setminus\overline{\theta}\right)  ^{3}$, smooth
and compactly supported, and integrating by parts with the help of the Neumann
boundary conditions (\ref{3.39}) lead to the integral identity%
\begin{equation}
\left(  AD\left(  \nabla_{\xi}\right)  v,D\left(  \nabla_{\xi}\right)
u\right)  _{\Lambda}=\left(  F,u\right)  _{\Lambda}+\sum\limits_{\pm}\left(
G^{\pm},u\right)  _{\Sigma^{\pm}}. \label{3.42}%
\end{equation}
Then the following result holds.

\begin{proposition}
\label{prop3.5}Under conditions (\ref{3.FG}), problem (\ref{3.38}%
)--(\ref{3.40}) has a unique weak solution $v\in V_{0}^{1}\left(
\Lambda;\theta\right)  ^{3}$ verifying the integral identity (\ref{3.42}) and
the norm $\left\Vert u;V_{0}^{1}(\Lambda;\theta)\right\Vert $, see
(\ref{3.41}), does not exceed the sum of norms of functions in (\ref{3.FG})
multiplied with a constant.
\end{proposition}

\begin{proof}
By virtue of the orthogonality condition in (\ref{3.FG}), the right-hand side
$\mathcal{F}(u)$ of (\ref{3.42}) can be rewritten as follows:%
\begin{align}
\mathcal{F}(u)  &  =\sum_{i=1}^{2}\left(  \left(  F_{i},u_{i}\right)
_{\Lambda}+\sum_{\pm}\left(  G_{i}^{\pm},u_{i}\right)  _{\Sigma^{\pm}}\right)
+\left(  F_{3}^{0},\overline{u}_{3}\right)  _{\mathbb{R}^{2}}\label{3.F}\\
&  +\left(  F_{3}^{1},u_{3}-\overline{u}_{3}\right)  _{\Lambda}+\sum_{\pm
}\left(  G_{3}^{\pm},u_{3}-\overline{u}_{3}\right)  _{\Sigma^{\pm}}\nonumber
\end{align}
where $\overline{u}_{3}\left(  \eta\right)  =\int_{-1/2}^{1/2}u_{3}\left(
\eta,\zeta\right)  d\zeta$ and $\left\Vert S_{2}\overline{u}_{3};L^{2}\left(
\mathbb{R}^{2}\right)  \right\Vert \leq\left\Vert S_{2}u_{3};L^{2}\left(
\Lambda\right)  \right\Vert $.

Since $\partial_{\zeta}u_{3}=\varepsilon_{33}\left(  u\right)  \in
L^{2}\left(  \Lambda\right)  $, the one-dimensional Poincar\'{e} and trace
inequalities in the interval $(-1/2,1/2)\ni\zeta$ integrated over the plane
$\mathbb{R}^{2}\ni\eta$, provide that%
\begin{align}
&  \left\Vert u_{3}-\overline{u}_{3};L^{2}(\Lambda)\right\Vert +\left\Vert
u_{3}-\overline{u}_{3};L^{2}\left(  \Sigma^{\pm}\right)  \right\Vert
\label{3.uu}\\
&  \leq c\left\Vert \partial_{z}\left(  u_{3}-\overline{u}_{3}\right)
;L^{2}\left(  \Lambda\right)  \right\Vert =c\left\Vert \partial_{z}u_{3}%
;L^{2}\left(  \Lambda\right)  \right\Vert \leq c\left\Vert u;V_{0}^{1}\left(
\Lambda;\theta\right)  \right\Vert .\nonumber
\end{align}
These together with the weighted trace inequality%
\[
\left\Vert S_{1}u_{i};L^{2}\left(  \Sigma^{\pm}\right)  \right\Vert \leq
c\left(  \left\Vert S_{1}\partial_{\zeta}u_{i};L^{2}\left(  \Lambda\right)
\right\Vert +\left\Vert S_{1}u_{i};L^{2}\left(  \Lambda\right)  \right\Vert
\right)  \leq c\left\Vert u;V_{0}^{1}\left(  \Lambda;\theta\right)
\right\Vert
\]
demonstrate that, under conditions (\ref{3.FG}), we have the continuous
functional (\ref{3.F}) on the right in (\ref{3.42}) if test functions are
taken from the space $V_{0}^{1}\left(  \Lambda;\theta\right)  ^{3}$. Since
inequality (\ref{3.K}) serves the left-hand side of (\ref{3.42}) to be a
scalar product in $V_{0}^{1}\left(  \Lambda;\theta\right)  ^{3}$, the Riesz
representation theorem proves the thesis.
\end{proof}

\bigskip

All elements of $V^{1}\left(  \Lambda\right)  $ (no condition on $\theta$ is
imposed) belong to $H_{loc}^{1}\left(  \overline{\Lambda}\right)  $. Any rigid
motion, except for the rotation $\xi_{1}e_{2}-\xi_{2}e_{1}$ which makes the
integral in (\ref{3.41}) divergent, falls into $V^{1}\left(  \Lambda\right)
^{3}$ because it can be approximated in norm (\ref{3.41}) by smooth compactly
supported vector functions (see \cite{na214, na243} for details). We introduce
the submatrix%
\begin{equation}
d^{\sharp}\left(  \eta,\zeta\right)  =\left(
\begin{array}
[c]{cccc}%
1 & 0 & 0 & \zeta\\
0 & 1 & -\zeta & 0\\
0 & 0 & \eta_{2} & -\eta_{1}%
\end{array}
\right)  \label{3.d}%
\end{equation}
of rigid motion and notice that the third and sixth columns are excluded from
the original matrix $d\left(  \eta,\zeta\right)  $ in (\ref{2.24}). Exhibiting
a general result in \cite{na243}, the next assertion detect%
\begin{equation}
v^{\sharp}(\eta,\zeta)=d^{\sharp}(\eta,\zeta)c^{\sharp},\ \ \ c^{\sharp}%
\in\mathbb{R}^{4}, \label{3.92}%
\end{equation}
as the main asymptotic term of the solution $v\in V_{0}^{1}(\Lambda
;\theta)^{3}$ under certain conditions on the right-hand sides.

\begin{proposition}
\label{prop3.6}Let the right-hand sides of problem (\ref{3.38})--(\ref{3.40})
satisfy the following smoothness and decay requirements%
\begin{equation}
\label{3.43}%
\begin{split}
\big|\nabla_{\eta}^{p}\partial_{\zeta}^{q}F_{i}\left(  \xi\right)
\big|+\big|\nabla_{\eta}^{p}G_{i}^{+}\left(  \eta\right)  \big|+\big|\nabla
_{\eta}^{p}G_{i}^{\bullet}\left(  \eta\right)  \big|  &  \leq c_{pq}%
\rho^{-2+\epsilon},\qquad i=1,2,\\
\big|\nabla_{\eta}^{p}\partial_{\zeta}^{q}F_{3}^{0}(\xi) \big|+\big|\nabla
_{\eta}^{p}G_{3}^{+}\left(  \eta\right)  \big|+\big|\nabla_{\eta}^{p}%
G_{3}^{\bullet}\left(  \eta\right)  \big|  &  \leq c_{pq}\rho^{-2+\epsilon},\\
\big|\nabla_{\eta}^{p}\partial_{\zeta}^{q}F_{3}^{1}\left(  \xi\right)  \big|
&  \leq c_{pq}\rho^{-3+\epsilon},\qquad\text{for }\rho\geq R_{\theta},
\end{split}
\end{equation}
where $p,q=0,1,2,\dots$, $\epsilon\in\left(  0,1\right)  $ and $R_{\theta}$ is
such that $\overline{\theta}\subset\mathbb{B}_{R_{\theta}}^{2}$. Then a
solution $v\in V_{0}^{1}\left(  \Lambda;\theta\right)  ^{3}$ of problem
(\ref{3.38})--(\ref{3.40}) given in Proposition \ref{prop3.5}, becomes smooth
for $\rho\geq R_{\theta}^{0}$ with any $R_{\theta}^{0}>R_{\theta}$ and
verifies the inequalities%
\begin{equation}
\big|\nabla_{\eta}^{p}\partial_{\zeta}^{q}\left(  v_{i}\left(  \xi\right)
-v_{i}^{0}\left(  \xi\right)  \right)  \big|\leq c_{pq}\rho^{\epsilon
},\ \ \ \big|\nabla_{\eta}^{p}\partial_{\zeta}^{q}\left(  v_{3}\left(
\xi\right)  -v_{3}^{0}\left(  \xi\right)  \right)  \big|\leq c_{pq}%
\rho^{1+\epsilon}\ \ \ \text{for }\rho\geq R_{\theta}^{0}, \label{3.44}%
\end{equation}
where $v^{0}$ is the rigid motion (\ref{3.92}) with a coefficient column
$c^{\sharp}\in\mathbb{R}^{4}$ depending on $F$, $G^{+}$, $G^{\bullet}$ and
other notation is the same as in (\ref{3.43}).
\end{proposition}

Note that the imposed orthogonality condition in (\ref{3.FG}) allowed us in
Proposition \ref{prop3.5} to reduce the decay requirement on $F_{3}^{0}$ and
$G_{3}^{\pm}$ because in the examination of functional (\ref{3.F}) we used
inequality (\ref{3.uu}) without referring to the weighted norm (\ref{3.41}).
Furthermore, these components can be compensated by a solution of a problem on
the interval $\left(  -1/2,1/2\right)  \ni\zeta$ which inherits the decay
properties from $F_{3}^{0}$ and $G_{3}^{+}$, $G_{3}^{\bullet}$ in (\ref{3.43})
and thus does not influence main asymptotic terms indicated in (\ref{3.44}).

The asymptotic form $v=v^{0}+\widetilde{v}$ is detected in \cite{na243} by
means of the dimension reduction, cf. Section \ref{sect3.1}, and an
application of the Kondratiev theory \cite{Ko} (see also \cite[Ch.3 and
Ch.6]{NaPl}) together with various weighted forms of Korn's inequality. The
paper \cite{na243} furnishes complete multi-scale decompositions of elastic
fields in a layer and we will precise the asymptotic forms of Proposition
\ref{prop3.6} in Section \ref{sect3.6}.

\subsection{The fundamental matrix of the differential operator $\mathcal{L}%
\left(  \nabla_{y}\right)  $\label{sect3.4}}

The most easily understood way to compensate for, e.g., the main part of the
discrepancy%
\[
h^{-1/2}w_{i}\left(  y\right)  e_{i}=h^{-1/2}w_{i}\left(  \mathcal{O}\right)
e_{i}+O(h^{1/2})\text{ \ \ for }y\in\theta_{h}%
\]
in condition (\ref{1.13}) is just to solve problem (\ref{3.38})--(\ref{3.40})
with $F=0$ and $G^{+}=0$, $G^{\bullet}=0$ but with the right-hand side
$-w_{i}\left(  \mathcal{O}\right)  e_{i}$ in the Dirichlet condition
(\ref{3.40}) on $\theta$. Since by an appropriate extension the inhomogeneity
can be passed over to the right-hand sides in equations (\ref{3.38}),
(\ref{3.39}), Proposition \ref{prop3.5} and our further comments on the rigid
motion (\ref{3.92}) demonstrate that the unique solution $v\in V^{1}\left(
\Lambda\right)  ^{3}$ of the problem is nothing but the constant vector
$-w_{i}\left(  \mathcal{O}\right)  e_{i}$ which does not decay as $\rho
\to+\infty$ and by no means can be accepted as a boundary layer. Thereupon,
instead of solutions offered by Propositions \ref{prop3.5} and \ref{prop3.6}
we prefer to employ a solution of the homogeneous problem (\ref{3.38}%
)--(\ref{3.40}) which, of course, must live outside $V_{0}^{1}\left(
\Lambda;\theta\right)  ^{3}$. Asymptotics at infinity of elastic fields with
power-logarithmic growth in a layer has been prepared in the paper
\cite{na243} and to achieve the goal we need a certain notation only.

General results in \cite{GelShilov} supply us with the fundamental matrix of
size $2\times2$%
\begin{equation}
\Phi^{\prime}\left(  y\right)  =\Psi^{\prime}\ln r+\psi^{\prime}\left(
\varphi\right)  \label{3.46}%
\end{equation}
of the elliptic $2\times2$-matrix (\ref{A9}) of the second-order differential
operators. Here, $\left(  r,\varphi\right)  $ is the polar coordinate system
on the plane $\mathbb{R}^{2}\ni y$, $\Psi^{\prime}$ is a numeral nondegenerate
symmetric $2\times2$-matrix and $\psi^{\prime}$ is a smooth matrix function on
the unit circle $\mathbb{S}$. By its meaning, the fundamental matrix
(\ref{3.46}) satisfies the relation%
\[
-\int_{\gamma}\mathcal{N}^{\prime}\left(  y,\nabla_{y}\right)  \Phi^{\prime
}\left(  y\right)  ds_{y}=\mathbb{I}_{2}\in\mathbb{R}^{2\times2}%
\]
where $\gamma$ is any smooth closed simple contour enveloping the
$y$-coordinates origin $\mathcal{O}$ and
\begin{equation}
\mathcal{N}^{\prime}\left(  y,\nabla_{y}\right)  =\mathcal{D}^{\prime}\left(
n\left(  y\right)  \right)  ^{\top}\mathcal{A}^{\prime}\mathcal{D}^{\prime
}\left(  \nabla_{y}\right)  \label{3.462}%
\end{equation}
is the Neumann condition operator for (\ref{A9}) with the unit vector
$n=\left(  n_{1},n_{2}\right)  ^{\top}$ of the outward normal at $\gamma$.
Since the matrix $\psi^{\prime}$ in (\ref{3.46}) is defined up to a constant
summand, it can be fixed such that%
\begin{equation}
\int_{\gamma}^{\prime}\Phi\left(  y\right)  ^{\top}\mathcal{N}^{\prime}\left(
y,\nabla_{y}\right)  \Phi^{\prime}\left(  y\right)  ds_{y}=\mathbb{O}_{2}%
\in\mathbb{R}^{2\times2}. \label{3.463}%
\end{equation}

The scalar forth-order elliptic operator (\ref{A99}) possesses the fundamental
solution $\Phi_{3}\left(  y\right)  $ which we write down in the convenient
form%
\begin{equation}
\Phi_{3}\left(  y\right)  =r^{2}\left(  -\frac12\Psi_{3}\ln r+\psi_{3}\left(
\varphi\right)  \right)  \label{3.45}%
\end{equation}
where $\Psi_{3}\neq0$ is a number and $\psi_{3}$ is a function on $\mathbb{S}%
$. Notice that $\Phi\left(  y\right)  =\operatorname*{diag}\left\{
\Phi^{\prime}\left(  y\right)  ,\Phi_{3}\left(  y\right)  \right\}  $ implies
the fundamental matrix of the operator $\mathcal{L}\left(  \nabla_{y}\right)
$ in (\ref{A9}).

In the isotropic case (\ref{1.10}) we have%
\[
\Phi^{\prime}\left(  y\right)  =\frac1{4\pi}\frac{\lambda^{\prime}+3\mu}%
{\mu\left(  \lambda^{\prime}+2\mu\right)  }\left(
\begin{array}
[c]{cc}%
-\ln r+\beta y_{1}^{2}r^{-2} & \beta y_{1}y_{2}r^{-2}\\
\beta y_{2}y_{1}r^{-2} & -\ln r+\beta y_{2}^{2}r^{-2}%
\end{array}
\right)  ,\ \ \ \Phi_{3}\left(  y\right)  =\frac3{8\pi}\frac{\lambda+2\mu}%
{\mu\left(  \lambda+\mu\right)  }r^{2}\ln r
\]
where $\beta=\left(  \lambda^{\prime}+3\mu\right)  ^{-1}\left(  \lambda
^{\prime}+\mu\right)  $, $\lambda$ and $\mu$ are the Lam\'{e} constants and
$\lambda^{\prime}$ is defined in (\ref{A3}).

We now introduce the $3\times4$-matrix%
\begin{equation}
\Phi^{\sharp}\left(  y\right)  =\left(  d^{\sharp}\left(  -\nabla
_{y},0\right)  ^{\top}\Phi\left(  y\right)  ^{\top}\right)  ^{\top}=\left(
\begin{array}
[c]{c}%
\Phi^{\prime}\left(  y\right) \\
\\
0\quad\quad0
\end{array}%
\begin{array}
[c]{cc}%
0 & 0\\
0 & 0\\
\Phi_{3}^{2}\left(  y\right)  & \Phi_{3}^{1}\left(  y\right)
\end{array}
\right)  \label{3.49}%
\end{equation}
where $\Phi^{\prime}\left(  y\right)  $ from (\ref{3.46}) has size $2\times2$
and%
\begin{equation}
\Phi_{3}^{i}\left(  y\right)  :=\frac{\partial\Phi_{3}}{\partial y_{i}}\left(
y\right)  =\Psi_{3}y_{i}\ln r+r\psi_{3}^{i}\left(  \varphi\right)
,\ \ \ i=1,2. \label{3.51}%
\end{equation}
To clarify properties of (\ref{3.45}) and (\ref{3.51}), we observe that
$\mathcal{D}_{3}\left(  \nabla_{y}\right)  =2^{-1/2}\mathcal{D}^{\prime
}\left(  \nabla_{y}\right)  \nabla_{y}$ in (\ref{A11}) and integrate by parts
as follows:%
\begin{equation}
\left(  \mathcal{L}_{3}w_{3},v_{3}\right)  _{\Gamma}+\left(  \mathcal{N}%
_{3}w_{3},\left(  1,\nabla_{y}\right)  v_{3}\right)  _{\gamma}=\left(
w_{3},\mathcal{L}_{3}v_{3}\right)  _{\Gamma}+\left(  \left(  1,\nabla
_{y}\right)  w_{3},\mathcal{N}_{3}v_{3}\right)  _{\gamma}. \label{3.452}%
\end{equation}
Here, $\Gamma$ is a domain bounded by the contour $\gamma$ (actually, we need
$\Gamma=\omega$ or $\Gamma=\mathbb{B}_{R}^{2})$ and $\mathcal{N}_{3}=\left(
\mathcal{N}_{0},\mathcal{N}_{1},\mathcal{N}_{2}\right)  $,%
\begin{equation}
\label{3.453}%
\begin{split}
\mathcal{N}_{0}\left(  y,\nabla_{y}\right)   &  =2^{-1/2}n\left(  y\right)
^{\top}\mathcal{D}^{\prime}\left(  -\nabla_{y}\right)  ^{\top}\mathcal{A}%
_{3}\mathcal{D}_{3}\left(  \nabla_{y}\right)  ,\\
\mathcal{N}_{i}\left(  y,\nabla_{y}\right)   &  =-2^{-1/2}e_{i}^{\top
}\mathcal{D}^{\prime}\left(  n\left(  y\right)  \right)  ^{\top}%
\mathcal{A}_{3}\mathcal{D}_{3}\left(  \nabla_{y}\right)  .
\end{split}
\end{equation}
The fundamental solution (\ref{3.45}) and its derivatives (\ref{3.51}) satisfy
the relations, with $i,k=1,2$
\begin{equation}
\label{3.454}%
\begin{split}
&  -\left(  1,\mathcal{N}_{0}\Phi_{3}\right)  _{\gamma}=1,\qquad-\left(
\left(  1,\nabla_{y}\right)  y_{k},\mathcal{N}_{3}\Phi_{3}\right)  _{\gamma
}=0,\\
&  -\left(  1,\mathcal{N}_{3}\Phi_{3}^{i}\right)  _{\gamma}=0,\qquad\left(
\left(  1,\nabla_{y}\right)  y_{k},\mathcal{N}_{3}\Phi_{3}^{i}\right)
_{\gamma}=\delta_{i,k}.
\end{split}
\end{equation}

\begin{remark}
\label{remPHIO}The simplest way to check up (\ref{3.454}) requires to use the
Dirac mass $\delta\left(  y\right)  $ in the framework of the theory of
distributions. Since $\mathcal{L}_{3}\left(  \nabla_{y}\right)  \Phi
_{3}\left(  y\right)  =\delta\left(  y\right)  $ by definition, we apply
formula (\ref{3.452}) with either $w_{3}=1$, $v_{3}=\Phi_{3}$ or $w_{3}=y_{k}%
$, $v_{3}=\Phi_{3}^{i}$ and obtain%
\[
-\left(  \left(  1,\nabla_{y}\right)  1,\mathcal{N}_{3}\Phi_{3}\right)
_{\gamma}=\left(  1,\delta\right)  _{\Gamma}=1,\ \ \ -\left(  \left(
1,\nabla_{y}\right)  y_{k},\mathcal{N}_{3}\Phi_{3}^{i}\right)  _{\gamma
}=\left(  y_{k},\partial\delta/\partial y_{i}\right)  _{\Gamma}=-\delta
_{i,k}.
\]
Other relations in (\ref{3.454}) as well as (\ref{3.462}), (\ref{3.463}) are
verified in the same way.
\end{remark}

\subsection{The elastic logarithmic capacity\label{sect3.5}}

According to a general procedure in \cite{na243} we search for a matrix
solution of the homogeneous problem (\ref{3.38})--(\ref{3.40}) in the form%
\begin{equation}
\mathcal{P}\left(  \xi\right)  =\left(  1-\chi\left(  \frac{\rho}{2R_{\theta}%
}\right)  \right)  \sum\limits_{p=0}^{3}W^{p}\left(  \zeta,\nabla_{\eta
}\right)  \Phi^{\sharp}\left(  \eta\right)  +\widehat{\mathcal{P}}\left(
\xi\right)  \label{3.54}%
\end{equation}
where the notation from (\ref{A12}) is used and $\widehat{\mathcal{P}}\in
V_{0}^{1}\left(  \Lambda;\omega_{1}\right)  ^{3\times4}$ is a remainder to be
determined. Note that $R_{\theta}$ is fixed such that $\overline{\theta
}\subset\mathbb{B}_{R_{\theta}}^{2}$.

We insert (\ref{3.54}) into the differential equations (\ref{3.38}) with $F=0$
and take decomposition (\ref{3.4}) into account. If $\rho>R_{\theta}$ and
$1-\chi\left(  \rho/2R_{\theta}\right)  =1$, we obtain for the sum $\Xi\left(
\eta,\zeta\right)  $ of the detached terms in (\ref{3.54}) that
\begin{equation}%
\begin{split}
D(-\nabla_{\eta})^{\top}AD(\nabla_{\eta})\Xi &  =L^{0}W^{0}\Phi^{\sharp
}+\left(  L^{0}W^{1}+L^{1}W^{0}\right)  \Phi^{\sharp}\\
&  \quad+\left(  L^{0}W^{2}+L^{1}W^{1}+L^{2}W^{0}\right)  \Phi^{\sharp
}+\left(  L^{0}W^{3}+L^{1}W^{2}+L^{2}W^{1}\right)  \Phi^{\sharp}\\
&  \quad+\left(  L^{1}W^{3}+L^{2}W^{2}\right)  \Phi^{\sharp}+L^{2}W^{3}%
\Phi^{\sharp}.
\end{split}
\label{3.55}%
\end{equation}
According to the content of Section \ref{sect3.2}, first four items on the
right in (\ref{3.55}) vanish; recall that $\mathcal{D}\left(  -\nabla_{\eta
}\right)  ^{\top}\mathcal{A}\mathcal{D}\left(  \nabla_{\eta}\right)
\Phi^{\sharp}\left(  \eta\right)  =0$ in the punctured plane $\mathbb{R}%
^{2}\setminus0$. Owing to (\ref{3.5}), (\ref{3.53}) and (\ref{3.46}),
(\ref{3.51}), the last two terms in (\ref{3.55}) are of order $\rho^{-3}$ and
$\rho^{-4}$ respectively.

In the Neumann boundary conditions (\ref{3.39}) we have%
\begin{equation}
\label{3.56}%
\begin{split}
&  D\left(  \pm e_{3}\right)  ^{\top}AD\left(  \nabla_{\eta}\right)
\Xi=N^{0\pm}W^{0}\Phi^{\sharp}+\left(  N^{0\pm}W^{1}+N^{1\pm}W^{0}\right)
\Phi^{\sharp}\\
&  \qquad+\left(  N^{0\pm}W^{2}+N^{1\pm}W^{1}\right)  \Phi^{\sharp}+\left(
N^{0\pm}W^{3}+N^{1\pm}W^{2}\right)  \Phi^{\sharp}+N^{1\pm}W^{3}\Phi^{\sharp}.
\end{split}
\end{equation}
First four items on the right vanish again. Recalling that the equation
(\ref{A17}) appeared as a result of calculation (\ref{A14}), we detect the
following relation which is crucial for our further consideration: for
$\eta\in\mathbb{R}^{2}\setminus0$%
\begin{equation}
\label{3.57}%
\begin{split}
&  e_{3}^{\top}\left(  \int_{-1/2}^{1/2}L^{1}\left(  \nabla_{\eta}%
,\partial_{\zeta}\right)  W^{3}\left(  \zeta,\nabla_{\eta}\right)
+L^{2}(\nabla_{\eta})W^{2}\left(  \zeta,\nabla_{\eta}\right)  d\zeta\right. \\
&  \qquad+\left.  \sum\limits_{\pm}N^{1\pm}\left(  \nabla_{\eta}\right)  W^{3}
\left(  \pm\frac12,\nabla_{\eta}\right)  \Phi^{\sharp}\left(  \eta\right)
\right)  =0.
\end{split}
\end{equation}

Let us compose a problem of type (\ref{3.38})--(\ref{3.40}) for the remainder
$\widehat{\mathcal{P}}$ in (\ref{3.54}). Since the cut-off function
$1-\chi\left(  2\rho/R_{\theta}\right)  $ is null near the $\eta$-coordinates
origin where $\Phi^{\sharp}$ becomes singular, the right-hand sides
$\widehat{F}$ and $\widehat{G}^{+}$, $\widehat{G}^{\bullet}$ in the problem
are smooth and, outside a big ball, coincide with (\ref{3.55}) and
(\ref{3.56}) respectively. Formula (\ref{3.57}) furnishes the representation
$\widehat{F}_{3}\left(  \xi\right)  =\widehat{F}_{3}^{0}\left(  \xi\right)
+\widehat{F}_{3}^{1}\left(  \eta\right)  $ for $\rho>2R_{\theta}$ where%
\begin{gather}
\widehat{F}_{3}^{1}\left(  \eta\right)  =e_{3}^{\top}\int_{-1/2}^{1/2}%
L^{2}\left(  \nabla_{\eta}\right)  W^{3}\left(  \zeta,\nabla_{\eta}\right)
d\zeta\Phi^{\sharp}\left(  \eta\right)  =O(\rho^{-4}),\nonumber\\
\widehat{F}_{3}^{0}\left(  \xi\right)  =e_{3}^{\top}\left(  L^{1}\left(
\nabla_{\eta},\partial_{\zeta}\right)  W^{3}\left(  \zeta,\nabla_{\eta
}\right)  +L^{2}\left(  \nabla_{\eta}\right)  W^{2}\left(  \zeta,\nabla_{\eta
}\right)  \right)  \Phi^{\sharp}\left(  \eta\right)  +\label{3.571}\\
+e_{3}^{\top}L^{2}\left(  \nabla_{\eta}\right)  W^{3}\left(  \zeta
,\nabla_{\eta}\right)  \Phi^{\sharp}\left(  \eta\right)  -\widehat{F}_{3}%
^{1}\left(  \eta\right)  =O(\rho^{-3}).\nonumber
\end{gather}
Note that%
\begin{equation}
G^{\pm}\left(  \eta\right)  =N^{1\pm}\left(  \nabla_{\eta}\right)
W^{3}\left(  \pm1/2,\nabla_{\eta}\right)  \Phi^{\sharp}\left(  \eta\right)
=O(\rho^{-3}),\ \ \ \rho>2R_{\theta}. \label{3.572}%
\end{equation}
Thus, all hypotheses in Proposition \ref{prop3.5} are fulfilled and then
remainder $\widehat{\mathcal{P}}\in V_{0}^{1}\left(  \Lambda;\omega
_{1}\right)  ^{3\times4}$ exists.

\bigskip

Formulas (\ref{3.571}) and (\ref{3.572}) also allow us to derive from
Proposition \ref{prop3.6} the representation%
\begin{equation}
\widehat{\mathcal{P}}\left(  \xi\right)  =d^{\sharp}\left(  \eta,\zeta\right)
C^{\sharp}+\widetilde{\mathcal{P}}\left(  \xi\right)  \label{3.58}%
\end{equation}
where the rows $\widetilde{\mathcal{P}}_{i}$, $i=1,2$, and $\widetilde
{\mathcal{P}}_{3}$ of the $3\times4$-matrix $\widetilde{\mathcal{P}}$ meet the
inequalities%
\begin{equation}
\big|\nabla_{\eta}^{p}\partial_{\zeta}^{q}\widetilde{\mathcal{P}}_{i}\left(
\xi\right)  \big|\leq c_{pq}\rho^{\epsilon},\ \ \ \big|\nabla_{\eta}%
^{p}\partial_{\zeta}^{q}\widetilde{\mathcal{P}}_{3}\left(  \xi\right)
\big|\leq c_{pq}\rho^{1+\epsilon}\ \ \ \text{for }\rho\geq R_{\theta}^{0},
\label{3.59}%
\end{equation}
with any $\epsilon\in\left(  0,1\right)  $.

We call the numerical matrix $C^{\sharp}=C^{\sharp}\left(  A,\theta\right)  $
of size $4\times4$ the \textit{elastic logarithmic capacity} (matrix), cf.
\cite{na135} for the isotropic case. It depends on the stifness matrix $A$ in
the Hooke law (\ref{1.9}) and the shape of the clamped zone $\theta$ and is
defined uniquely through decompositions (\ref{3.54}), (\ref{3.58}) of the
\textit{elastic logarithmic potential} matrix $\mathcal{P}\left(  \xi\right)
$ of size $3\times4$.

\begin{remark}
\label{remLOG}The names introduced above come by analogy with the logarithmic
capacity potential $P\left(  \eta\right)  $ in harmonic analysis, cf.
\cite{Land, PoSe}. The function $P$ is harmonic in $\mathbb{R}^{2}%
\setminus\overline{\vartheta}$, vanishes at the boundary $\partial\vartheta$
of a compact set $\vartheta$ and admits the representation%
\[
P\left(  \eta\right)  =\left(  2\pi\right)  ^{-1}\left(  \ln\rho^{-1}+C_{\log
}\left(  \vartheta\right)  \right)  +O(\rho^{-1}),\ \ \ \rho\to+\infty,
\]
while $-\left(  2\pi\right)  ^{-1}\ln\left\vert \eta\right\vert $ is the
fundamental solution of the Laplacian in the plane $\mathbb{R}^{2}$ and
$C_{\log}\left(  \vartheta\right)  $ is called the logarithmic capacity.
Clearly, $P$ also solves the mixed boundary value problem in the perforated
layer $\Theta=\left\{  \xi:\eta\in\mathbb{R}^{2}\setminus\overline{\vartheta
},\ \left\vert \zeta\right\vert <1/2\right\}  $, namely%
\[
-\Delta_{\xi}P\left(  \xi\right)  =0,\ \xi\in\Theta,\ \ \ P\left(  \xi\right)
=0,\ \xi\in\partial\vartheta\times\left(  -1/2,1/2\right)  ,\ \ \ \pm
\partial_{\zeta}P\left(  \eta,\pm1/2\right)  =0,\ \eta\in\mathbb{R}%
^{2}\setminus\overline{\vartheta}%
\]
which really looks quite similar to a scalar version of the elasticity problem
(\ref{3.38})--(\ref{3.40}) under consideration. Note that, as distinct from
the standard harmonic capacity in dimension $n\geq3$ which always is positive,
the logarithmic capacity $C_{\log}\left(  \vartheta\right)  $ can be positive
(the obstacle $\vartheta$ is big) or negative ($\vartheta$ is small).
\end{remark}

The elastic logarithmic capacity matrix $C^{\sharp}$, in general, is neither
positive, nor negative but still symmetric. The latter is proved in
\cite{na135} for an isotropic layer $\Lambda$ with a defect and we here serve
for a much more complicated anisotropic case.

\begin{theorem}
\label{th3.7}The elastic logarithmic capacity $4\times4$-matrix $C^{\sharp
}=C^{\sharp}\left(  A,\theta\right)  $ is symmetric.
\end{theorem}

\begin{proof}
We insert on both positions in the Green formula for the operator $L\left(
\nabla_{\xi}\right)  $ in the truncated layer $\Lambda\left(  T\right)
=\left\{  \xi\in\Lambda:\rho<T\right\}  $ and let $T\rightarrow+\infty$. Since
$\mathcal{P}$ verifies the homogeneous problem (\ref{3.38})--(\ref{3.40}), we
are left with an integral over the truncation surface $\mathbb{S}_{T}%
\times\left(  -1/2,1/2\right)  $, the lateral surface of the circular cylinder
$\Lambda\left(  T\right)  $. We have%
\begin{equation}
0=\int_{\mathbb{S}_{T}}\int_{-1/2}^{1/2}\left(  \mathcal{P}(\xi)^{\top}%
D(\rho^{-1}\eta,0)^{\top}AD\left(  \nabla_{\xi}\right)  \mathcal{P}%
(\xi)-\left(  D(n(\eta),0)^{\top}AD\left(  \nabla_{\xi}\right)  \mathcal{P}%
(\xi)\right)  ^{\top}\mathcal{P}(\xi)\right)  d\zeta ds_{\eta} \label{8.1}%
\end{equation}
where $n\left(  \eta\right)  =\rho^{-1}\eta$ is the unit vector of the outward
normal and $ds_{\eta}=\rho d\varphi$ is the arc element on the circle
$\mathbb{S}_{T}$.

We now extract from the integrand in (\ref{8.1}) all terms of order $\rho
^{-1}$ which contribute to the limit. Infinitesimal terms $o\left(  \rho
^{-1}\right)  $, in particular ones generated by the remainder $\widetilde
{\mathcal{P}}$ in (\ref{3.58}) can be removed from the integrand with a clear
reason while terms growing as $\rho\to+\infty$ vanish all together after
integration in $\left(  \varphi,\zeta\right)  \in\left(  0,2\pi\right)
\times\left(  -1/2,1/2\right)  $ because the limit does exist. In this way,
the integral in $\zeta$ is approximated by the expression $J\left(
\eta\right)  -J\left(  \eta\right)  ^{\top}$ where, in view of the obvious
relation $D\left(  \nabla_{\xi}\right)  d^{\sharp}(\xi)=0$, we have
\begin{equation}
J(\eta)=\int_{-1/2}^{1/2}\bigg(\sum_{q=0}^{3}W^{q}\left(  \zeta,\nabla_{\eta
}\right)  \Phi^{\sharp}(\eta) +d^{\sharp}(\eta,\zeta)C^{\sharp}\bigg)^{\top
}D\left(  \frac{\eta}{\rho},0\right)  ^{\top}AD\left(  \nabla_{\xi}\right)
\sum\limits_{p=0}^{3}W^{p}\left(  \zeta,\nabla_{\eta}\right)  \Phi^{\sharp
}(\eta)d\zeta. \label{8.2}%
\end{equation}
Similarly to (\ref{3.55}) and (\ref{3.56}), in the formula%
\begin{align*}
D\left(  \nabla_{\xi}\right)  \sum\limits_{p=0}^{3}W^{p}\Phi^{\sharp}  &
=D_{\zeta}W^{0}\Phi^{\sharp}+\left(  D_{\zeta}W^{1}+D_{\eta}W^{0}\right)
\Phi^{\sharp}+\left(  D_{\zeta}W^{2}+D_{\eta}W^{1}\right)  \Phi^{\sharp}+\\
&  +\left(  D_{\zeta}W^{3}+D_{\eta}W^{2}\right)  \Phi^{\sharp}+D_{\eta}%
W^{3}\Phi^{\sharp},
\end{align*}
where $D_{\zeta}=D\left(  0,0,\partial_{\zeta}\right)  $ and $D_{\eta
}=D\left(  \nabla_{\eta},0\right)  $, cf. (\ref{3.5}), the first two terms on
the right vanish.

According to (\ref{3.40}) and (\ref{3.46}), (\ref{3.51}), we have%
\[
(D_{\zeta}W^{1+p}\left(  \zeta,\nabla_{\eta}\right)  +D_{\eta}W^{p}\left(
\zeta,\nabla_{\eta}\right)  )\Phi^{\sharp}\left(  \eta\right)  =O(\rho
^{-p}),\ p=1,2\ \ \ D_{\eta}W^{3}\left(  \zeta,\nabla_{\eta}\right)
\Phi^{\sharp}\left(  \eta\right)  =O(\rho^{-3}).
\]
Taking (\ref{3.d}) and (\ref{3.53}) into account, we write the relations
\[
\left(  W^{0}+W^{1}\left(  \zeta,\nabla_{\eta}\right)  \right)  \Phi^{\sharp
}+d^{\sharp}\left(  \eta,\zeta\right)  C^{\sharp}=d^{\flat}\left(
\zeta,\nabla_{\eta}\right)  \left(  \Phi^{\sharp}\left(  \eta\right)
+d^{\sharp}\left(  \eta,0\right)  C^{\sharp}\right)
\]
and%
\[
D\left(  n\left(  \eta\right)  ,0\right)  d^{\flat}\left(  \zeta,\nabla_{\eta
}\right)  =D^{1\flat}\left(  n\left(  \eta\right)  ,\zeta,\nabla_{\eta
}\right)  +D^{0\flat}\left(  n\left(  \eta\right)  \right)
\]
with the matrices $D^{0\flat}\left(  n\right)  =\left(  D_{\zeta}\zeta\right)
\left(  n^{\top},0\right)  ^{\top}e_{3}$ of size $6\times3$ and%
\[
D^{1\flat}\left(  n,\zeta,\nabla_{\eta}\right)  =\left(
\begin{array}
[c]{c}%
\mathcal{D}^{\prime}\left(  n\right)  ,\ \ -\zeta2^{1/2}\mathcal{D}^{\prime
}\left(  n\right)  \nabla_{\eta}\\
\mathbb{O}_{3}%
\end{array}
\right)  ,\ \ \ d^{\flat}\left(  \zeta,\nabla_{\eta}\right)  =\left(
\begin{array}
[c]{ccc}%
1 & 0 & -\zeta\partial/\partial\eta_{1}\\
0 & 1 & -\zeta\partial/\partial\eta_{2}\\
0 & 0 & 0
\end{array}
\right)  .
\]
Comparing decay rates of remaining multipliers shows that we need to keep in
(\ref{8.2}) the following two terms of order $\rho^{-1}\left(  1+\left\vert
\ln\rho\right\vert \right)  $:%
\begin{equation}
\int_{-1/2}^{1/2}D^{1\flat}\left(  n(\eta),\zeta,\nabla_{\eta}\right)  \left(
\Phi^{\sharp}(\eta)+d^{\sharp}\left(  \eta,0\right)  C^{\sharp}\right)
^{\top}A\left(  D_{\zeta}W^{2}\left(  \zeta,\nabla_{\eta}\right)  +D_{\eta
}W^{1}\left(  \zeta,\nabla_{\eta}\right)  \right)  d\zeta\Phi^{\sharp}(\eta)
\label{8.5}%
\end{equation}
and%
\begin{equation}
\int_{-1/2}^{1/2}D^{0\flat}\left(  n\left(  \eta\right)  \right)  \left(
\Phi^{\sharp}\left(  \eta\right)  +d^{\sharp}\left(  \eta,0\right)  C^{\sharp
}\right)  ^{\top}A\left(  D_{\zeta}W^{3}\left(  \zeta,\nabla_{\eta}\right)
+D_{\eta}W^{2}\left(  \zeta,\nabla_{\eta}\right)  \right)  d\zeta\Phi^{\sharp
}\left(  \eta\right)  . \label{8.6}%
\end{equation}

By formulas (\ref{A16}), (\ref{A17}) for $\mathcal{L}^{\prime}\left(
\nabla_{y}\right)  $, $\mathcal{L}_{3}\left(  \nabla_{y}\right)  $ and
(\ref{3.462}), (\ref{3.453}) for $\mathcal{N}^{\prime}$, $\mathcal{N}_{j}$ we
recall the calculations (\ref{A14}), (\ref{A15}) and conclude that (\ref{8.5})
is equal to the sum%
\begin{equation}
\label{8.61}%
\begin{split}
&  \left(  \left(  \Phi^{\sharp}\left(  \eta\right)  \right)  ^{\prime
}+\left(  d^{\sharp}\left(  \eta,0\right)  \right)  ^{\prime}C^{\sharp
}\right)  ^{\top}\mathcal{N}^{\prime}\left(  \eta,\nabla_{\eta}\right)
\left(  \Phi^{\sharp}\left(  \eta\right)  \right)  ^{\prime}\\
&  \qquad+\left(  \nabla_{\eta}\left(  \Phi_{3}^{\sharp}\left(  \eta\right)
\right)  ^{\prime}+d_{3}^{\sharp}\left(  \eta,0\right)  C^{\sharp}\right)
\left(
\begin{array}
[c]{c}%
\mathcal{N}_{1}\left(  \eta,\nabla_{\eta}\right) \\
\mathcal{N}_{2}\left(  \eta,\nabla_{\eta}\right)
\end{array}
\right)  \Phi_{3}^{\sharp}\left(  \eta\right)
\end{split}
\end{equation}
where $\left(  \Phi^{\sharp}\right)  ^{\prime}$ and $\left(  d^{\sharp
}\right)  ^{\prime}$ are submatrices with eliminated third lines. A similar
argument converts (\ref{8.6}) into%
\begin{equation}
\left(  \Phi_{3}^{\sharp}\left(  \eta\right)  +d_{3}^{\sharp}\left(
\eta,0\right)  C^{\sharp}\right)  ^{\top}\mathcal{N}_{0}\left(  \eta
,\nabla_{\eta}\right)  \Phi_{3}^{\sharp}\left(  \eta\right)  . \label{8.62}%
\end{equation}
Now we see that%
\begin{equation}
0=\lim_{T\to+\infty}\int_{\mathbb{S}}\left(  J\left(  \eta\right)  -J\left(
\eta\right)  ^{\top}\right)  ds_{\eta}, \label{8.7}%
\end{equation}
i.e. the limit of the integral on the right-hand side of (\ref{8.1}), is a
linear combination of scalar products, listed in (\ref{3.463}) and
(\ref{3.454}), and their conjugates. Thus, a part of the integral (\ref{8.7})
which involves columns of the matrix $\Phi^{\sharp}$, vanishes due to
"orthogonality" conditions in (\ref{3.463}) and (\ref{3.454}) with $0$ on the
right-hand side. The other part of the integral which involves columns of the
matrices $d^{\sharp}$ and $\Phi^{\sharp}$, converts into the difference
$\left(  C^{\sharp}\right)  ^{\top}-C^{\sharp}$. Indeed, we had reduced
expression (\ref{8.2}) to the sum of (\ref{8.61}) and (\ref{8.62}) so that
after integration in $\varphi\in\mathbb{S}$ it is sufficient to apply the
"bi-orthogonality" conditions (\ref{3.463}) and (\ref{3.454}) with $1$ on the
right-hand side.

Since according to (\ref{8.1}) the integral in (\ref{8.7}) is null, the
theorem is proved.
\end{proof}

\subsection{Asymptotics at infinity of elastic fields in layer-shaped
domains\label{sect3.6}}

If the right-hand sides $F$ and $G^{+}$, $G^{\bullet}$ in problem
(\ref{3.38})--(\ref{3.40}) are smooth and decay exponentially as $\rho
\to+\infty$, for example have compact supports, then results in \cite{na243}
serve for an asymptotic expansion of the solution $v$ with a remainder of any
given power-law decay $O\left(  \rho^{-N}\right)  $. Hence, we can make the
decomposition of the elastic logarithmic potential $\mathcal{P}$ a bit more
precise, namely%
\begin{equation}
\mathcal{P}\left(  \xi\right)  =\left(  1-\chi(2\rho/R_{\theta})\right)
\left(  \sum_{p=0}^{3}W^{p}\left(  \zeta,\nabla_{\eta}\right)  \Phi^{\sharp
}\left(  \eta\right)  +d^{\sharp}\left(  \eta,0\right)  C^{\sharp}%
+\Upsilon^{\sharp}\left(  \eta\right)  \right)  +\widetilde{\widetilde
{\mathcal{P}}}\left(  \xi\right)  . \label{8.8}%
\end{equation}
The first two terms, of course, are the same as in (\ref{3.54}) and
(\ref{3.58}), but the new remainder gets the faster decay properties, cf.
(\ref{3.59}),%
\[
\left\vert \nabla_{\eta}^{p}\partial_{\zeta}^{q}\widetilde{\widetilde
{\mathcal{P}}}_{i}\left(  \xi\right)  \right\vert \leq c_{pq}\rho
^{-1+\epsilon},\ \ \ \left\vert \nabla_{\eta}^{p}\partial_{\zeta}%
^{q}\widetilde{\widetilde{\mathcal{P}}}_{3}\left(  \xi\right)  \right\vert
\leq c_{pq}\rho^{\epsilon},\ \ \ \rho\geq2R_{\theta}.
\]
The latter is caused by the additional term $\Upsilon^{\sharp}\left(
\eta\right)  $ in (\ref{8.8}), a $3\times4$-matrix with the rows
\begin{equation}
\label{8.10}%
\begin{split}
&  \Upsilon_{i}^{\sharp}\left(  \eta\right)  =\rho^{-1}\Upsilon_{i}^{0\sharp
}\left(  \varphi\right)  +\rho^{-1}\ln\rho\Upsilon_{i}^{1\sharp}\left(
\varphi\right)  ,\qquad i=1,2,\\
&  \Upsilon_{3}^{\sharp}\left(  \eta\right)  =\Upsilon_{1}^{0\sharp}\left(
\varphi\right)  +\ln\rho~\Upsilon_{3}^{1\sharp}\left(  \varphi\right)
+\left(  \ln\rho\right)  ^{2}\Upsilon_{3}^{2\sharp}\left(  \varphi\right)  .
\end{split}
\end{equation}
Let us explain where the lower-order asymptotic terms (\ref{8.10}) appear from.

The dimension reduction procedure, quite similar to Section \ref{sect3.1},
leads to the following system for the third term $\Upsilon^{\sharp}$ in the
asymptotic ansatz (\ref{8.8}):%
\begin{equation}
\mathcal{D}\left(  -\nabla_{\eta}\right)  ^{\top}\mathcal{A}\mathcal{D}\left(
\nabla_{\eta}\right)  \Upsilon^{\sharp}\left(  \eta\right)  =\mathcal{F}%
\left(  \eta\right)  ,\ \ \ \eta\in\mathbb{R}^{2}\setminus\mathcal{O}.
\label{8.11}%
\end{equation}
The right-hand side $\mathcal{F}=\left(  \mathcal{F}_{1},\mathcal{F}%
_{2},\mathcal{F}_{3}\right)  ^{\top}$ is constructed in the same way as
(\ref{A14}), (\ref{A15}) and is necessary to compensate for discrepancies of
the expression $\Xi$, see a comment to (\ref{3.54}), in the homogeneous
equations (\ref{3.38}) and (\ref{3.39}), that are the fifth and sixth terms in
(\ref{3.55}) and the fifth term in (\ref{3.56}) respectively. In this way,
representations (\ref{3.53}), (\ref{3.49}) and integration by parts in the
variable $\zeta$, cf. calculation in Section \ref{sect3.5}, yield
\begin{equation}
\label{8.12}%
\begin{split}
&  \mathcal{F}_{i}(\eta)=e_{i}^{\top}D_{\eta}\int_{-1/2}^{1/2}A\left(
D_{\zeta}W^{3}\left(  \zeta,\nabla_{\eta}\right)  +D_{\eta}W^{2}\left(
\zeta,\nabla_{\eta}\right)  \right)  d\zeta\Phi^{\sharp}(\eta) =\rho
^{-3}\mathcal{F}_{i}^{0}(\varphi),\\
&  \mathcal{F}_{3}(\eta)=e_{i}^{\top}D_{\eta}\int_{-1/2}^{1/2}AD_{\eta}%
W^{3}\left(  \zeta,\nabla_{\eta}\right)  d\zeta\Phi^{\sharp}(\eta)\\
&  \qquad+\sum\limits_{i=1}^{2} \frac{\partial}{\partial\eta_{i}}D_{\eta}%
\int_{-1/2}^{1/2}\zeta A\left(  D_{\zeta}W^{3}\left(  \zeta,\nabla_{\eta
}\right)  +D_{\eta}W^{2}\left(  \zeta,\nabla_{\eta}\right)  \right)
d\zeta\Phi^{\sharp}(\eta) =\rho^{-4}\mathcal{F}_{3}^{0}(\varphi).
\end{split}
\end{equation}
It is important that logarithms figuring in the matrix $\Phi^{\sharp}\left(
\eta\right)  $ due to (\ref{3.46}), (\ref{3.51}) and (\ref{3.49}) as
coefficients on linear functions in $\eta_{i}$, are eliminated in (\ref{8.12})
by differentiating sufficiently many times so that functions (\ref{8.12})
become positive homogeneous in $\eta$ of degree $-3$ and $-4$ respectively.
However, the logarithm $\ln\rho$ returns into components (\ref{8.10}) of the
solution $\Upsilon^{\sharp}\left(  \eta\right)  $ by virtue of the Kondratiev
theorem on asymptotics (see \cite{Ko} and, e.g., \cite[Thm. 3.1.4]{NaPl}). To
confirm this, we observe system (\ref{8.11}) with $\mathcal{F}=0$ and find out
its vector solutions in the form of derivatives of columns in matrix
(\ref{3.51}), namely%
\begin{equation}
\sum_{i=1}^{2}\left(
\begin{array}
[c]{c}%
\left(  a_{1}^{i}\partial_{1}+a_{2}^{i}\partial_{2}\right)  \Phi_{i}^{\prime
}\left(  \eta\right) \\
\left(  a_{1}^{2+i}\partial_{1}+a_{2}^{2+i}\partial_{2}\right)  \Phi_{3}%
^{i}\left(  \eta\right)
\end{array}
\right)  . \label{8.13}%
\end{equation}
As shown in \cite[\S 5.4]{NaPl}, any solution with the same positive
homogeneity degrees $-1$ and $0$ as in (\ref{8.10}) takes form (\ref{8.13}).
Furthermore, in the cases $a_{2}^{3}\neq0$ and $a_{1}^{3}\neq0$ the third
component of (\ref{8.13}) stays linearly dependent on $\ln\rho$. Hence, the
above-mentioned theorem prescribes to search for a solution of system
(\ref{8.11}) with the right-hand sides (\ref{8.12}) in the form (\ref{8.10})
while $\Upsilon_{1}^{1\sharp}$, $\Upsilon_{2}^{1\sharp}$ and $\Upsilon
_{3}^{2\sharp}$ may vanish only under some orthogonality conditions in
$L^{2}\left(  \mathbb{S}\right)  ^{3}$ for the angular part $\mathcal{F}%
^{0}\left(  \varphi\right)  =\left(  \mathcal{F}_{1}^{0}\left(  \varphi
\right)  ,\mathcal{F}_{2}^{0}\left(  \varphi\right)  ,\mathcal{F}_{3}%
^{0}\left(  \varphi\right)  \right)  ^{\top}$ of $\mathcal{F}\left(
\eta\right)  $. We do not examine these conditions and specify (\ref{8.10})
because the component $\Upsilon$ is indicated in (\ref{8.8}) with an auxiliary
technical reason only and will be excluded from the final asymptotic formulas
for the solution $u$ of problem (\ref{1.11})--(\ref{1.14}) in the paper
\cite{ButCaNa2}.

We emphasize that a solution $\Upsilon^{\sharp}$ of system (\ref{8.11}) is
defined up to a summand of type (\ref{8.13}). A unique $\Upsilon^{\sharp}$ in
decomposition (\ref{8.8}) of the elastic logarithmic potential $\mathcal{P}$
depends on the whole data of the problem, in particular on the clamped area
$\theta$, and is specified by means of theorem on asymptotics in layer-shaped
domains, the most challenging assertion in the paper \cite{na243}.

\end{document}